\documentclass[pra, a4paper,10pt, showpacs, twocolumn]{revtex4}

\usepackage{bm, bbm, amsmath, amssymb, amsthm}
\usepackage[T1]{fontenc}
\usepackage[latin9]{inputenc}
\usepackage{graphicx}
\usepackage{geometry}
\theoremstyle{plain}
\newtheorem{thm}{Theorem}
\theoremstyle{plain}
\newtheorem{lem}[thm]{Lemma}
\theoremstyle{definition}
\newtheorem{example}[thm]{Example}
\theoremstyle{plain}
\newtheorem*{thm*}{Theorem}
\theoremstyle{definition}
\newtheorem{defn}[thm]{Definition}
\theoremstyle{remark}
\newtheorem{rem}[thm]{Remark}
\theoremstyle{plain}
\newtheorem{cor}[thm]{Corollary}

\newcommand{\one}{\mathbbm{1}}

\newcommand{\ket}[1]{\left|{#1}\right\rangle}
\newcommand{\bra}[1]{\left\langle{#1}\right|}
\newcommand{\braket}[2]{\langle{#1}|{#2}\rangle}
\newcommand{\ketbrad}[1]{\left|{#1}\rangle\!\langle{#1}\right|}

\begin{document}

\title{Classical spin systems and the quantum stabilizer formalism: general
mappings and applications}

\author{R. Hübener$^{1}$, M. Van den Nest$^{3}$, W. Dür$^{1,2}$ and
H. J. Briegel$^{1,2}$}

\affiliation{$^{1}$ Institut für Theoretische Physik, Universität Innsbruck,
Technikerstraße 25, A-6020 Innsbruck, Austria\\
$^{2}$ Institut für Quantenoptik und Quanteninformation der Österreichischen
Akademie der Wissenschaften, Technikerstraße 21, A-6020 Innsbruck, Austria\\
$^{3}$ Max-Planck-Institut für Quantenoptik, Hans-Kopfermann-Str.
1, D-85748 Garching, Germany.}

\begin{abstract}
We present general mappings between classical spin systems and quantum physics. More precisely, we show how to express partition functions and correlation functions of arbitrary classical spin models as inner products between quantum stabilizer states and product states, thereby generalizing mappings for some specific models established in our previous work [Phys. Rev. Lett. 98, 117207 (2007)]. For Ising- and Potts-type models with and without external magnetic field, we show how the entanglement features of the corresponding stabilizer states are related to the interaction pattern of the classical model, while the choice of product states encodes the details of interaction. These mappings establish a link between the fields of classical statistical mechanics and quantum information theory, which we utilize to transfer techniques and methods developed in one field to gain insight into the other. For example, we use quantum information techniques to recover well known duality relations and local symmetries of classical models in a simple way, and provide new classical simulation methods to simulate certain types of classical spin models. We show that in this way all inhomogeneous models of $q$-dimensional spins with pairwise interaction pattern specified by a graph of bounded tree-width can be simulated efficiently. Finally, we show relations between classical spin models and measurement-based quantum computation.
\end{abstract}

\pacs{03.67.-a,03.67.Lx,75.10.Hk,75.10.Pq,02.70.-c}

\maketitle


\section{Introduction}

\emph{Classical spin systems} are widely studied in statistical physics \cite{Wu84}. They also play an important role in modeling complex behavior also in other disciplines, such as economics and biology. In spite of their often simple definition, spin models show a highly non-trivial behavior, as is, e.g., apparent from their phase structure and criticality. Surprisingly, even the simple Ising model of interacting 2-state spins arranged on a 2D square lattice (with external magnetic field) is in general not solvable, and calculating, e.g., the ground state energy or the partition function is known to be a computationally hard problem~\cite{Ba82}.

In \emph{quantum information theory} (QIT), on the other hand, (entanglement) properties of quantum systems are systematically studied, and possible applications regarding, e.g., quantum computation are investigated. QIT has become a field of interdisciplinary interest, and concepts and methods developed in QIT have found applications also in other branches of physics. In the context of QIT, methods to efficiently compute and simulate certain quantum systems and their properties have been developed~\cite{SDV06,VDB07}. In particular, so-called quantum ``stabilizer states''~\cite{Go97,HDM05,VPhD} and ``graph states''~\cite{Hei05, SW00} have been introduced and studied in detail. Stabilizer states are used for certain types of quantum error correction~\cite{Go97} and measurement-based quantum computation~\cite{RB01}, and can be described efficiently in terms of their stabilizing operators. This allows to determine many of their (entanglement) properties, and to efficiently simulate some processes classically.

In this paper, we present general mappings between classical spin systems and quantum physics related to QIT. More precisely, we show how to express the partition function and correlation functions of an arbitrary classical spin system as a quantum mechanical amplitude (scalar product) between a stabilizer state $|\psi\rangle$ encoding the interaction pattern, and a certain product state $\otimes_{j}|\alpha_{j}\rangle$ encoding the details of the interaction (i.e. the coupling strengths) and the temperature:
\begin{equation}
Z_{G}=\bra{\psi}\left(\bigotimes_{j}\ket{\alpha_{j}}\right).
\label{eq:Zeqoverlap}
\end{equation}
With such a mapping at hand, we can use methods and techniques established in one field to gain insight into the other, thereby providing a novel approach to these problems. We have initiated this approach in a recent publication~\cite{VDB06}, where such mappings have been established for Ising and Potts-type models. Here we generalize this approach, and discuss the mappings and their applications in more detail.

We further note that connections between quantum information theory and statistical mechanics have recently been studied by several other researchers~\cite{many,BR06}.

\subsection{Mappings between classical spin systems and quantum physics}

In this section we briefly sketch the general form of the proposed mappings between classical and quantum systems.

We consider classical $q$-state spin systems with an arbitrary pairwise interaction pattern, described by a graph $G$ with vertex set $V$ (position of the classical spins) and edge set $E$ (corresponding to interactions). Such systems are sometimes called ``edge models'' (i.e., the interactions take place on the \emph{edges}). Each spin $s$ may assume $q$ different states: $s\in\{0, \dots, q-1\}$. We will consider models where the pairwise interactions $h(s,s')$ between spins $s$ and $s'$ are of the following forms:
\begin{itemize}
\item[(i)] $h(s, s')$ only depends on the difference (modulo $q$) of
the two involved spins, $h\equiv h(|s -s'|_{q})$;
\item[(ii)] $h(s, s')$ is
of the form (i), but with additional local magnetic fields;
\item[(iii)] $h(s, s')$
is completely arbitrary.
\end{itemize}
We will also consider (iv) models with arbitrary $k$-body interactions.

The Ising- and Potts model without [with] magnetic field are of type (i) [(ii)] respectively, while so-called ``vertex models'' (i.e., the interactions take place on the vertices) are a special case of type (iv).

In each of the cases (i)-(iv), we show how one can express the \emph{partition function} $Z_{G}$ as an overlap between a quantum stabilizer state and a  complete product state, (Eq.~\eqref{eq:Zeqoverlap}). Depending on the different  forms of the interaction (as in (i)-(iv)), these quantum states will be defined slightly differently.
\begin{itemize}
\item[(i)] For models \emph{without} local fields, the corresponding quantum states consists of $|E|$ $q$-level quantum systems (one for each pairwise interaction term). We will denote the stabilizer state by $|\psi_{G}\rangle$. The product state has the form $|\alpha\rangle= \bigotimes_{e\in E}|\alpha_{e}\rangle$, where the coefficients of each $|\alpha_e\rangle$ encode the strengths of the pairwise couplings, as well as the temperature of the system.
\item[(ii)] For models \emph{with} local magnetic fields, the corresponding quantum states consist of $|V|+|E|$ $q$-level quantum systems (one for each pairwise interaction term and one for each local field), with stabilizer states denoted by $|\varphi_{G}\rangle$ and a product state $|\alpha\rangle=\bigotimes_{e\in E}|\alpha_{e}\rangle\bigotimes_{a\in V}|\alpha_{a}\rangle$.
\item[(iii)] For models with \emph{general} pairwise interaction (iii), we provide a mapping where the stabilizer state is a tensor product of $|V|$ entangled states, $|\phi\rangle=\bigotimes_{a\in V}(\sum_{j=0}^{q-1} |j\rangle^{\otimes n_{a}})$. Here, $n_{a}$ is the degree of vertex $a$, i.e.\ the number of neighbors in the graph, which also determines the number of associated $q$-level quantum systems. Correspondingly, we now consider states $|\alpha_{ab}\rangle$ of dimension $q^{2}$ for the overlap, which are associated to one quantum particle belonging to vertex $a$ and and one quantum particle belonging to vertex $b$. A similar picture holds for models with arbitrary $k$-body interactions (iv), where the product states have now dimension $q^{k}$, and are associated with multiple vertices.
\end{itemize}

We will investigate the entanglement properties of the states $|\psi_{G}\rangle$ and $|\varphi_{G}\rangle$ and their relation to the underlying interaction pattern specified by the graph $G$, and provide a number of examples to illustrate this connection.

The mappings (ii)-(iv) can be extended, and will allow us to express also \emph{classical correlation functions} in a quantum language.

\subsection{Applications of the mappings}

Based on these mappings, we will then illustrate some applications. Here we briefly sketch which applications can be obtained.
\begin{itemize}
\item[(a)] Using well established stabilizer methods \cite{Hei05,Go97,HDM05,VPhD}, we show how one can recover the well known high-low temperature duality relations~\cite{Wu84} for classical spin models on arbitrary planar graphs.
\item[(b)] Using the fact that stabilizer states are stabilized by certain tensor product operators, we derive local symmetry relations for classical models, i.e.\ we identify models with different coupling strengths that lead to the same partition function.
\item[(c)] We show how one can use recently established results in QIT to classically simulate certain classes of quantum systems efficiently~\cite{SDV06, MS06, VDB07} and thus obtain novel simulation algorithms for classical spin system. More precisely, by describing stabilizer states in terms of an optimal tree tensor network~\cite{SDV06} of dimension $d$, one can compute the overlap with product states with an effort that is polynomial in $d$. This leads to an efficient algorithm to classically simulate arbitrary (inhomogeneous) classical $q$-state models on graphs with a bounded (or logarithmically growing) tree width. We also extend these results to models with $k$-body interaction.
\item[(d)] Finally, we discuss links between classical spin models and \emph{measurement based quantum computation}. This allows us to relate the computational complexity of computing partition functions of classical spin models with the quantum computational power of the associated graph states.

\end{itemize}
We also note that (d) has recently been used in Ref.~\cite{CDBip} to show a ``\emph{completeness}'' property of the 2D Ising model. That is, invoking the connection to measurement-based quantum computation, it was shown that the partition function of any model with pairwise interaction in arbitrary dimension can be expressed as a special instance of the partition function of a 2D Ising model on an (enlarged) 2D square lattice (with complex coupling strengths).

\subsection{Guideline through the paper}

The paper is organized as follows. We start in Sec.~\ref{sec:Background-on-spin} by briefly reviewing classical spin models, and collect some relevant results on stabilizer and graph states in Sec.\ \ref{sec:StabilizerStates}. We then introduce different mappings between classical spin systems and quantum mechanical amplitudes, and discuss the properties of the involved quantum states in Secs.~\ref{sec:Encoding-classical-spin} and \ref{sec:Extending-the-formalism}. We illustrate a number of applications of these mappings in Sec.\ \ref{sec:Applications}, and summarize and conclude in Sec.\ \ref{sec:Summary-and-Conclusion}.

\section{Background on spin models\label{sec:Background-on-spin}}

In this section we describe the classical models that we want to consider. Since the various approaches to be described later are related and can be viewed as derivations from an original scheme, we will focus on the original approach first.

The typical model to be considered by the original approach is the \emph{thermal state of a classical spin model described by a Hamiltonian function with two-body interaction}, and this model will serve as an introductory guide to the general idea. These systems have the virtue that they admit a description by means of a graph~\cite{GR01}: the spins of the system correspond to the vertices and the two-body interaction pattern between the spins is given by the edge set.

We will describe a mapping of such an interaction graph to a stabilizer state of a quantum system. Performing an overlap of this quantum stabilizer state with another quantum product state, encoding the temperature and individual interaction strengths, then yields the properties of the thermal state of the classical system. We want to emphasize that this evaluation is \emph{not approximate but exact}. Later on, extensions of this formalism will be given as well, going beyond this particular kind of graphical description and at the same time going beyond the limitation to two-body interactions.

It is important to keep in mind that the interaction \emph{pattern} and the interaction \emph{strengths} are encoded at different places: the graph encodes the interaction pattern, not the strengths, hence an edge connecting two vertices simply denotes the fact that there is an interaction taking place. The strength and nature of this interaction is not encoded in the graph, but in a product state to be specified later.  This encoding admits the strengths of all edge terms and all vertex terms to be chosen individually, hence the interaction strength may vary for different pairs of spins and also the local field may vary.

More precisely, let, for now, $H$ be a Hamiltonian function with two-body-interaction between classical spins $s$ that can assume $q$ possible values $s\in\left\{ 0,...,q-1\right\} $. In the graphical description of this Hamiltonian function, we let $G=\left(V,E\right)$ denote the graph associated with $H$, where the sets $V$ and $E$ contain the vertices and the edges of the graph respectively. In this picture, any vertex $v \in V$ corresponds to a classical spin site $s_v$ and any edge $e\in E$ between adjacent vertices $v_{1},v_{2}$ of the graph corresponds to an interaction term between the respective spins $s_{v_{1}}$ and $s_{v_{2}}$. Additionally, we  allow each spin $s_{v}$ to contribute a local term to the Hamiltonian function, i.e.\ a term that that depends on the state of the site $s_{v}$ alone, although this is not reflected in the graph. We might think of the energy of the spin in a local field. We choose the graph to be a directed one, denoting the orientation by $\sigma$. The exact choice of the directions can be arbitrary but has to be fixed. This way, the two adjacent vertices of an edge $e\in E$ can be distinguished as ``head'' $v_{e}^{+}$ and ``tail'' $v_{e}^{-}$ of the edge, respectively.

We will derive several different mappings for Hamiltonian functions described by these graphs. The first mapping admits descriptions of systems with classical Hamiltonian functions of the form
\begin{equation}
H\left(\left\{ s_{i}\right\} \right)
=\sum_{e\in E}h_{e}\bigl(\bigl|s_{v_{e}^{+}}-s_{v_{e}^{-}}\bigr|_{q}\bigr),
\label{eq:first_approach}
\end{equation}
with $h_{e}$ being an energy term that depends on the relative state of two interacting spins $s_{v_{e}^{+}}$ and $s_{v_{e}^{-}}$ modulo $q$. In the second mapping we extend the quantum description to be able to include also external fields
\begin{equation}
H\left(\left\{ s_{i}\right\} \right)
=\sum_{e\in E}h_{e}\bigl(\bigl|s_{v_{e}^{+}}-s_{v_{e}^{-}}\bigr|_{q}\bigr)
+\sum_{v\in V}b_{v}\bigl(s_{v}\bigr),\label{eq:second_approach}
\end{equation}
where $b_{v}$ is an energy term contributed by a local external field, acting on the spin $s_{v}$. To go beyond the limitation to interaction Hamiltonian functions that depend on the relative state of the spins only, we finally provide further approaches to treat Hamiltonian functions of the form
\[
H\left(\left\{ s_{i}\right\} \right)=\sum_{\left(ij\right)\in E}
h_{\left(ij\right)}\bigl(s_{i},s_{j}\bigr)
\]
as well as arbitrary Hamiltonian functions with $n$-body terms.

The degrees of freedom in the definitions of these Hamiltonian functions give rise to a large set of possible classical spin systems to be described -- even if we restricted ourselves to the sets of Hamiltonian functions specified in Eqns.~\eqref{eq:first_approach} and \eqref{eq:second_approach}. Among those are the Ising model, the Potts model and the clock model on arbitrary lattices, all equipped with (local) magnetic fields, and generalizations thereof~\cite{Wu84}.

\subsection{Ising model}

The Ising model describes a set of classical spins (or simply dipoles) that can point either up or down and are placed on a graph. All next neighbors have the same distance (hence the interaction strength is uniformly given by the real number $J$) and long-range forces  are neglected. Moreover, there is a global external field whose strength is given by the real number $B$, which puts an energetic bias on the possible configurations. Thus the classical Ising model is described by the Hamiltonian function
\begin{equation}
H_{\mbox{\tiny{Ising}}}\left(\left\{ s_{i}\right\} \right)=
J\sum_{\bigl\langle i,j\bigr\rangle}\bigl|s_{i}-s_{j}\bigr|_{2}+B\sum_{i}
\bigl(s_{i}-\frac{1}{2}\bigr),\label{eq:Ising, our form}
\end{equation}
where the $s_{i}\in\left\{ 0,1\right\} $ and $\bigl\langle i,j\bigr\rangle$ denotes that $i$ and $j$ are adjacent spins on the graph. We note that it can be rewritten as
\[
H_{\mbox{\tiny{Ising}}}\left(\left\{ \sigma_{i}\right\} \right)=-J'\sum_{\bigl\langle i,j\bigr\rangle}\sigma_{i}\sigma_{j}+B'\sum_{i}\sigma_{i},
\]
where $\sigma_{i}\in\left\{ +1,-1\right\}$. This is the more familiar form and  can be obtained from Eq.~\eqref{eq:Ising, our form} by a rescaling of parameters and an addition of a constant. Although this model is highly idealized, it features (in appropriate dimensions) many properties of realistic solids, such as phase transitions, spontaneous symmetry breaking etc. As will be shown, our treatment allows---without a change of computational effort---the generalization to spin-glass Hamiltonian functions, where the factor $J$ is actually dependent on the specific pairs of spins that interact: $J\rightarrow J_{ij}$.

\subsection{Potts and clock models}

A generalization of the Ising model is given by the Potts- and the clock model. Whereas the individual spins in the Ising model can take only one of two values and hence for neighbors there are only the alternatives of being \emph{parallel} or \emph{anti-parallel}, it might be desirable to allow the individual dipoles to assume more positions and hence to obtain more relative configurations of neighbors that can be discriminated energetically. Accordingly we choose spin states $s_{i}\in\left\{ 0,...,q-1\right\} $ and a Hamiltonian function
\begin{equation}
H\left(\left\{ s_{i}\right\} \right)=
-\sum_{\bigl\langle i,j\bigr\rangle}J\left(\Theta_{ij}\right)+b\sum_{i}
\bigl(s_{i}-\frac{q-1}{2}\bigr),\label{eq:HamPottsClock}
\end{equation}
where $\Theta_{ij}$ is a function that discriminates the relative states of neighboring spins. We can interpret it for instance as the angle between adjacent dipoles, provided that they can only rotate in a fixed plane, e.g., $\Theta_{ij}=\Theta_{i}-\Theta_{j}$ with discretised positions $\Theta_{i}=2\pi s_{i}/q$. The function $J$, which characterizes the Hamiltonian function, maps the relative angle (i.e., relative state) of adjacent spins to an energy value: The Potts model is defined by
\[
J_{\mbox{\tiny{Potts}}}(\Theta_{ij}):=-\varepsilon\delta(\Theta_{ij})
\]
with $\varepsilon\in\mathbb{R}$ and the clock model by
\[
J_{\mbox{\tiny{clock}}}(\Theta_{ij}):=-\varepsilon\cos(\Theta_{ij}).
\]

\subsection{Partition function}

The focus of this paper will be on the thermal equilibrium of these classical systems. More precisely, the central quantities of interest that we want to obtain are the partition function
\[
Z\left(\beta\right)=\sum_{\left\{ s_{i}\right\} }
e^{-\beta H\left(\left\{ s_{i}\right\} \right)}
\]
as well as the n-point correlation functions, whose definition can be found,  e.g., in Ref.~\cite{ME86}
\begin{multline*}
\left\langle s_{i_{1}},s_{i_{2}},...,s_{i_{n}}\right\rangle _{\beta}\\
=Z^{-1}\sum_{\left\{ s_{i}\right\}}\cos\left(\Theta_{i_{1}}\right)
\cos\left(\Theta_{i_{2}}\right)...\cos\left(\Theta_{i_{n}}\right)
e^{-\beta H\left(\left\{ s_{i}\right\} \right)}.
\end{multline*}
The partition function encodes the macroscopic properties of a thermal ensemble. The parameters that enter depend on the kind of ensemble we look at, e.g., the canonical (temperature), grand canonical (temperature and chemical potential) and others. In the present framework we will deal with the canonical ensemble, because the number of spin sites is fixed, but energy can be drawn from an external bath.

Let us briefly illustrate the importance of the partition function. The partition function of a canonical ensemble is
\[
Z=\sum_{i}e^{-\beta E_{i}},
\]
where the index $i$ is the index for the states with energy $E_{i}$ that the system can take and $\beta=\left(k_{B}T\right)^{-1}$ with the Boltzmann constant $k_{B}$. Moreover, $p_{i}=Z_{i}^{-1}e^{-\beta E_{i}}$ is the probability to find the system in the state with energy $E_{i}$. Several relevant quantities can now be derived from $Z$: We can extract the expectation value of the energy
\[
\left\langle E\right\rangle _{\beta}=Z^{-1}\sum_{i}E_{i}e^{-\beta E_{i}}
=-\frac{\partial\log Z}{\partial\beta},
\]
the variance of the expected energy
\[
\left\langle \left(\delta E\right)^{2}\right\rangle _{\beta}
=\frac{\partial^{2}\log Z}{\partial\beta^{2}},
\]
as well as the free energy
\[
F=\left\langle E\right\rangle _{\beta}-TS=-\beta^{-1}\log Z,
\]
where the entropy is $S=-k_{B}\sum_{i}p_{i}\log p_{i}$, and more. We refer the reader to standard text books on this topic.

\section{\label{sec:StabilizerStates}Stabilizer states and graph states}

In this section, we give the definition and some properties of stabilizer states~\cite{Go97,HDM05,VPhD} and graph states~\cite{Hei05,SW00}.  We will first consider spin-1/2 quantum systems, then proceed to higher  dimensional systems.

\subsection{Graph states}

Here we will briefly familiarize the reader with the graph states.  In the present context, a graph $G=\left(V,E\right)$ is identified with a quantum system. Each vertex $a$ represents a quantum spin, and the adjacent vertices (connected with $a$ by edges in the graph) form the neighborhood $N_{a}$ of $a$. This way, the graph defines a set of operators
\[
K_{a}:=\sigma_{x}^{\left(a\right)}\prod_{b\in N_{a}}\sigma_{z}^{\left(b\right)},
\]
where the sigma-matrices are defined as usual\begin{gather}
\sigma_{0}=\left(\begin{array}{cc}
1 & 0\\
0 & 1\end{array}\right),\sigma_{x}=\left(\begin{array}{cc}
0 & 1\\
1 & 0\end{array}\right),\label{eq:sigma-matrices}\\
\sigma_{y}=\left(\begin{array}{cc}
0 & -i\\
i & 0\end{array}\right),\sigma_{z}=\left(\begin{array}{cc}
1 & 0\\
0 & -1\end{array}\right),\nonumber \end{gather}
and the notation $O^{\left(a\right)}$ of an operator $O$ means the tensor product of the operator $O$, acting on the subspace of site $a$, and $\one$ everywhere else. A graph state $\left|G\right\rangle $ associated with to the graph $G$, and hence with the set $\left\{ K_{a}\right\} $, is the unique non-trivial fixed point of the operators $K_{a}$,
\[
\forall K_{a}:K_{a}\left|G\right\rangle =\left|G\right\rangle.
\]

Graph states are a subset of the stabilizer states, which play an important role in the context of one-way quantum computing. Conversely, every stabilizer state can be written, up to a local rotation, as a graph state.

\subsection{Stabilizer states}

We will now turn our attention to the slightly more general set of stabilizer states. The concept of defining a state as a simultaneous fixed point of a set of operators can be used in a slightly more general way than in the case of graph states, where the operators $K_{a}$ take a very special form. To construct more general sets of operators we consider the sigma-matrices, see formula ~\eqref{eq:sigma-matrices}, and the group they generate
\[
\mathcal{G}_{1}=\left\{ \pm\sigma_{0},\pm i\sigma_{0},\pm\sigma_{x},
\pm i\sigma_{x},\pm\sigma_{y},\pm i\sigma_{y},\pm\sigma_{z},\pm i\sigma_{z}\right\} .
\]
Tensor products of $\mathcal{G}_{1}$ with itself form the Pauli groups $\mathcal{G}_{n}:=\mathcal{G}_{1}^{\otimes n}$. It is known that any Abelian subgroup $\mathcal{S}\subset\mathcal{G}_{n}$ of a Pauli group with $\bigl|\mathcal{S}\bigr|=2^{n}$ that does not contain $-\one_{n}$ has a unique fixed point $\left|\psi\right\rangle $ in the Hilbert space $\mathcal{H}$ that it acts upon. We then call $\mathcal{S}$ the \emph{stabilizer} of $\left|\psi\right\rangle $ and $\left|\psi\right\rangle $ a \emph{stabilizer state.} It should be noted that each stabilizer can be identified with its generator, i.e., a set of operators that generate it. Generators are not unique sets, but share the necessary requirement to contain $n$ independent operators.

For our purposes, the prefactor $\left(\pm1,\pm i\right)$ of an element of a Pauli group will not be important. Moreover, there is a mapping between the Pauli group $\mathcal{G}_{n}/_{\sim}$ ($\mathcal{G}_{n}$ modulo prefactors) and the group $\mathbb{F}_{2}^{2n}$, which will be used later. Since $\sigma_{y}=i\sigma_{x}\sigma_{z}$ and $\sigma^{0}=\one_{2}$ for all sigma-matrices, we can encode the generators of $\mathcal{G}_{1}/_{\sim}$ as follows
\begin{align*}
\sigma_{0} & \sim\sigma_{x}^{0}\sigma_{z}^{0}\mapsto(00)\\
\sigma_{x} & \sim\sigma_{x}^{1}\sigma_{z}^{0}\mapsto(10)\\
\sigma_{y} & \sim\sigma_{x}^{1}\sigma_{z}^{1}\mapsto(11)\\
\sigma_{z} & \sim\sigma_{x}^{0}\sigma_{z}^{1}\mapsto(01).
\end{align*}
where $\sim$ denotes equality modulo prefactor. Tensor products of these operators and hence elements of the groups $\mathcal{G}_{n}/_{\sim}$ will be encoded by the mapping
\[
\mathcal{G}_{n}/_{\sim}\ni\bigotimes_{i=1}^{n}
\sigma_{x}^{\xi_{i}}\sigma_{z}^{\zeta_{i}}\mapsto\left(\xi_{1},...,\xi_{n},
\zeta_{1}...,\zeta_{n}\right)\in\mathbb{F}_{2}^{2n}.
\]

The generalization to $q$-dimensional quantum systems with  $\mathcal{H}=\left(\mathbb{C}^{2}\right)^{\otimes q}$ is straightforward. We replace $\sigma_{x}$ and $\sigma_{z}$ by the operators $X$ and $Z$ respectively, where
\[
X\left|j\right\rangle =\left|j+1\text{ mod }q\right\rangle ,\quad Z\left|j\right\rangle
=e^{2\pi ij/q}\left|j\right\rangle ,
\]
$q=2$ being a special case that gives us back $\sigma_{x}$ and $\sigma_{z}$. The higher-dimensional groups $\mathcal{G}_{n}^{q}/_{\sim}$ are thus generated by tensor products of $X^{a}Z^{b}$ where $a,b=0,...,q-1.$ The mapping is generalized to the group homomorphism
\begin{multline*}
\left(\mathcal{G}_{n}^{q}/_{\sim},
\cdot\right)\ni\bigotimes_{i=1}^{n}X^{\xi_{i}}Z^{\zeta_{i}}\\
\mapsto\left(\xi_{1},...,\xi_{n},\zeta_{1}...,\zeta_{n}\right)
\in\left(\mathbb{F}_{q}^{2n},+\right).
\end{multline*}

The number of elements in a stabilizer that stabilizes one single stabilizer state is $q^{n}$, the number of elements of its generator is $n$.

Related to this construction is a theorem that we will use later. Note that we do not neglect the phase this time.

\begin{lem}
\label{lem:commutation}Any two operators  $\bigotimes_{i=1}^{n}X^{\xi_{i}}Z^{\zeta_{i}}$ and
\textup{$\bigotimes_{i=1}^{n}X^{\xi'_{i}}Z^{\zeta'_{i}}$} commute
if and only if $\xi'\cdot\zeta-\xi\cdot\zeta'=0$ modulo $q$.
\end{lem}
\begin{proof}
The computation for the single spin site yields
\begin{multline*}
X^{\xi_{i}}Z^{\zeta_{i}}X^{\xi'_{i}}Z^{\zeta'_{i}}
=X^{\xi_{i}+\xi_{i}'}Z^{\zeta_{i}+\zeta_{i}'}e^{2\pi i\xi'_{i}\zeta_{i}/q}\\
=X^{\xi_{i}'}Z^{\zeta'_{i}}X^{\xi_{i}}Z^{\zeta_{i}}
e^{2\pi i(\xi'_{i}\zeta_{i}-\xi_{i}\zeta_{i}')/q}.
\end{multline*}
Hence for all sites together we obtain a phase factor  $e^{2\pi i(\xi'\cdot\zeta-\xi\cdot\zeta')/q}$.
\end{proof}

It is noteworthy that for $q=2$ each stabilizer state is related to a graph state by some local unitary transformations. This means that the two sets do not differ as far as their non-local properties are concerned.

The stabilizer states are interesting to us, because---as will be shown---the interaction patterns of the Hamiltonian functions of the classical spin systems that we look at correspond to such states. Moreover, stabilizer states are well investigated and elaborate techniques are known for their manipulation~\cite{Go97}, allowing us to investigate relationships between different (interaction) graphs and hence different Hamiltonian functions.

\section{\label{sec:Encoding-classical-spin}Encoding classical spin systems
in quantum language}

In this section we will investigate in detail the correspondence of the classical and the quantum systems that were presented in the preceding sections.

\subsection{The basic principle\label{sub:The-basic-principle}}

The basic approach, which was introduced in Ref.~\cite{VDB06}, is sufficient to describe systems with classical Hamiltonian functions of the form  $H\left(\left\{ s_{i}\right\} \right)=\sum_{e\in E}h_{e}\bigl(\bigl| s_{v_{e}^{+}}-s_{v_{e}^{-}}\bigr|_{q}\bigr)$. The idea is to map the graph $G$, describing the interaction pattern into a stabilizer state, together with a supplementary product state that encodes the interaction strengths as well as the temperature.

Let the classical spin system be defined by the (arbitrarily oriented)  interaction graph $G^{\sigma}=\left(V,E\right)$ over $\bigl|V\bigr|$ classical spins of dimension $q$, where $\sigma$ denotes the orientation.  Let in the following $M=\bigl|V\bigr|$ and $N=\bigl|E\bigr|$. Now consider the incidence matrix $B^{\sigma}$ of the interaction graph $G^{\sigma}$. This matrix has one row for each vertex and one column for each edge. The entries are either $0$ or $\pm1$, where $B_{v,e}^{\sigma}=-1$ if the index pair $\left(v,e\right)$ corresponds to the tail vertex $v$ of edge $e$, $B_{v,e}^{\sigma}=+1$ for the head vertex $v$ of edge $e$ and $B_{v,e}^{\sigma}=0$ otherwise. Consistent with our notation, we do not consider graphs with edges that connect one point with itself. The rows of $B^{\sigma}$ span the $\mathbb{Z}_{q}$-vector space $C_{G}\left(q\right)$, which is a linear subspace of $\mathbb{Z}_{q}^{N}$. The vectors $(B^{\sigma})^Ts\in C_{G}\left(q\right)$, where $T$ denotes transposition, correspond to the vectors that encode spin configurations $\left(s_{v}\right)_{v\in V}$, as the linear mapping 
\[
\bigl|s_{v_{e}^{+}}-s_{v_{e}^{-}}\bigr|_{q}
=\left|\sum_{v\in V}\left(B^{\sigma}\right)_{e,v}^Ts_{v}\right|_{q}
\]
shows.

\begin{lem}
\label{lem:prefactor}The kernel of the linear mapping $\left(B^{\sigma}\right)^T$ has $q^{\kappa}$ elements, where $\kappa$ is the number of connected sub-graphs of $G$ (without isolated points).
\end{lem}
\begin{proof}
We re-arrange the rows of the matrix of $\left(B^{\sigma}\right)^T$ so that the connected sub-graphs $G_{i}=\left(E_{i},V_{i}\right)$ are described by blocks $B_{i}$, i.e.,
\[
\left(B^{\sigma}\right)^T\mapsto\left(\begin{array}{cccc}
B_{1} & 0 & 0 & \cdots\\
0 & B_{2} & 0 & \cdots\\
0 & 0 & B_{3} & \cdots\\
\vdots & \vdots & \vdots & \ddots\end{array}\right).
\]
Within each connected sub-graph $G_{i}$, there is at least one path from each vertex $v$ to each other vertex $v'$: $\left(v,v_{0},v_{1},...,v'\right)$, each edge $\left(v_{n},v_{n+1}\right)$ in this path being represented by one row in the corresponding matrix $B_{i}$. Since a vector $\mathbf{s}$ to be in the kernel of $B_{i}$ implies $\left|s_{v_{n}}-s_{v_{n+1}}\right|_{q}=0$ for each edge $\left(v_{n},v_{n+1}\right)$, we deduce immediately that $\left|s_{v}-s_{v'}\right|_{q}=0$ for any two vertices $v$ and $v'$ in $V_{i}$. Hence if $\mathbf{s}$ is in the kernel of $B_{i}$, all spins in $\left\{ s_{v_{n}}\right\} _{v_{n}\in V_{i}}$ take the same value. So there are $q$ different vectors in the kernel of each matrix $B_{i}$, of which there are $\kappa$.
\end{proof}

We are now ready to define an non-normalized stabilizer state encoding $G^{\sigma}$. We obtain it by first interpreting each vector  $c=\left(c_{1},c_{2},...,c_{N}\right)\in C_{G}\left(q\right)$ as a product state of a multipartite quantum spin system with spin-dimensionality $q$ according to the formula $\left|c\right\rangle  :=\left|c_{1}\right\rangle \otimes\left|c_{2}\right\rangle  \otimes...\otimes\left|c_{N}\right\rangle $ and by a subsequent summation of all these states~\cite{footnote}
\begin{equation}
\left|\psi_{G}\right\rangle
:=q^{\kappa}\sum_{c\in C_{G}\left(q\right)}\left|c\right\rangle =\sum_{\mathbf{s}\in\mathbb{Z}_{q}^{N}}\bigl|\left(B^{\sigma}\right)^T\mathbf{s}\bigr\rangle,
\label{eq:def-of-psiG}
\end{equation}
where the second equality and the factor $q^{\kappa}$ follow immediately from lemma (\ref{lem:prefactor}). For an illustrative example see Fig.~\ref{fig:nondecorated}.
\begin{figure*}
\includegraphics[width=1.2\columnwidth]{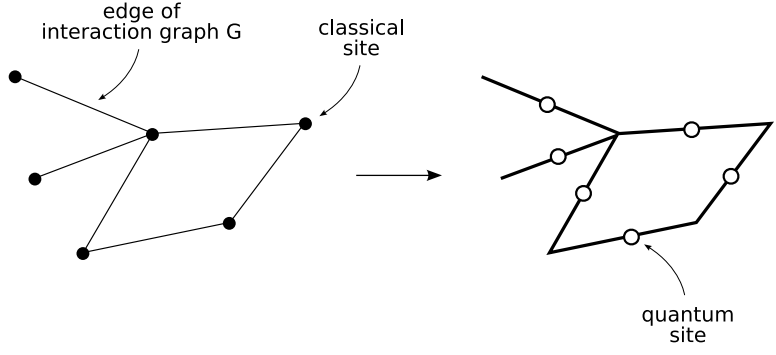}

\caption{\emph{The basic construction principle.} This figure shows an example of an encoding of a classical interaction pattern into a stabilizer state. Thin graph: the classical interaction graph G; thick graph: the derived graph relating quantum sites in a stabilizer state. The classical spin sites correspond to vertices in a graph G. The interacting pairs of sites are mapped to a quantum site, one for each edge (``edge qudits''). The quantum sites form, by construction, a stabilizer state.}

\label{fig:nondecorated}
\end{figure*}

\begin{lem}
\label{lem:stabilizer-state}The state $\left|\psi_{G}\right\rangle $ is a stabilizer state. Its stabilizer consists of the $q^{N}$ operators
\begin{equation}
X\left(v\right)Z\left(u\right)
:=\bigotimes_{e\in E}X^{v_{e}}Z^{u_{e}},\label{eq:XvZu}
\end{equation}
where $v\in C_{G}\left(q\right)$ and $u\in C_{G}\left(q\right)^{\perp}.$
\end{lem}
\begin{proof} From Lemma (\ref{lem:commutation}) we derive immediately that, by using the given construction rule for the operators, we obtain a commuting set. Considering the equation 
\[
X^{\xi_{i}}Z^{\zeta_{i}}X^{\xi'_{i}}Z^{\zeta'_{i}}
=X^{\xi_{i}+\xi_{i}'}Z^{\zeta_{i}+\zeta_{i}'}e^{2\pi i\xi'_{i}\zeta_{i}/q}
\]
and hence 
\[
X\left(v\right)Z\left(u\right)X\left(v'\right)Z\left(u'\right)
=X\left(v+v'\right)Z\left(u+u'\right)e^{2\pi iu \cdot v' /q},
\]
where $u \cdot v'=\sum_i u_i v'_i =0$ for each admissible choice of these vectors, we also see that these operators form a group. Furthermore, these operators actually stabilize the (nontrivial) state $\left|\psi_{G}\right\rangle $, since for all $v\in C_{G}\left(q\right)$ and for all $u\in C_{G}\left(q\right)^{\perp}$
\begin{align*}
X\left(v\right)Z\left(u\right)\left|\psi_{G}\right\rangle  & 
=X\left(v\right)Z\left(u\right)q^{\kappa}\sum_{c\in C_{G}\left(q\right)}\left|c\right\rangle \\
&=q^{\kappa}\sum_{c\in C_{G}\left(q\right)}X\left(v\right)Z\left(u\right)\left|c\right\rangle \\
&=q^{\kappa}\sum_{c\in C_{G}\left(q\right)}e^{2\pi i u \cdot c/q}\left|c+v\right\rangle \\
&=q^{\kappa}\sum_{c'\in C_{G}\left(q\right)}\left|c'\right\rangle =\left|\psi_{G}\right\rangle .
\end{align*} From this, we can moreover deduce that $-\one$ is not element of this set of operators.

For this set to be a stabilizer of a single state of our Hilbert space, the number of elements in this set has to be $q^{N}$. This part of the proof is given in Appendix \ref{sec:NumElProof}.
\end{proof}

\subsubsection{Thermal quantities}

Now we are able to formulate the central theorem of this section.

\begin{thm}
The partition function $Z_{G}\left(q,\left\{ h_{e}\right\} \right)$ of a classical spin system defined on the graph $G=\left(V,E\right)$ by the Hamiltonian function \textup{$H\left(\left\{ s_{i}\right\} \right)=\sum_{e\in E}h_{e}\bigl(\bigl|s_{v_{e}^{+}}-s_{v_{e}^{-}}\bigr|_{q}\bigr)$} can be written as the overlap of a stabilizer state and a product state
\[
Z_{G}\left(q,\left\{ h_{e}\right\} \right)=
(\bigotimes_{e\in E}\bra{\alpha_{e}})\ket{\psi_{G}}
\]
of a quantum mechanical spin-system, where
\begin{equation}
\ket{\alpha_{e}}
=\sum_{j=0}^{q-1}e^{-\beta h_{e}\left(j\right)}\ket{j}.
\label{eq:def-of-alpha}
\end{equation}

\end{thm}
\begin{proof}
The state $\left|\psi_{G}\right\rangle $ is a stabilizer state according to lemma (\ref{lem:stabilizer-state}), and we compute, with an arbitrarily chosen orientation $\sigma$ of the graph $G$,
\begin{eqnarray*}
(\bigotimes_{e\in E} \bra{\alpha_{e}}) \ket{\psi_{G}}
& \overset{(\ref{eq:def-of-psiG})}{=} & \sum_{\mathbf{s}\in\mathbb{Z}_{q}^{N}}
(\bigotimes_{e\in E} \bra{\alpha_{e}})\ket{\left(B^{\sigma}\right)^T
\mathbf{s}} \\
& \overset{(\ref{eq:def-of-alpha})}{=} & \sum_{\mathbf{s}\in\mathbb{Z}_{q}^{N}}
\prod_{e\in E} e^{-\beta h_{e}(\bigl|s_{v_{e}^{+}}-s_{v_{e}^{-}}\bigr|_{q})}\\
& = &\sum_{\mathbf{s}\in\mathbb{Z}_{q}^{N}}e^{-\beta H\left(\left\{ s_{i}\right\} \right)}
\end{eqnarray*}
which concludes the proof.
\end{proof}

Let us give a brief interpretation of the method used to encode the partition function. We observe that to calculate partition functions of systems with Hamiltonian functions of the form
\[
H\left(\left\{ s_{i}\right\} \right)
=\sum_{e\in E}h_{e}\bigl(\bigl|s_{v_{e}^{+}}-s_{v_{e}^{-}}\bigr|_{q}\bigr)
\]
(zero external field) it is already sufficient to know the relative state of spins whose corresponding vertices are connected by an edge. Accordingly, we map each vector $\left(s_{v}\right)_{v\in V}$ of spin configurations to the corresponding one $(B^{\sigma})^T s = \left(\bigl|s_{v_{e}^{+}}-s_{v_{e}^{-}}\bigr|_{q}\right)_{e\in E}$ of differences along edges using the incidence matrix $B^{\sigma}$. These vectors are automatically consistent with spin-configurations, and moreover, there can be no more of them than we have already given. 

As shown, an interaction pattern is encoded into a graph and this graph is encoded into a stabilizer state. Furthermore, the corresponding interaction strengths (as well as a temperature) are encoded into a product state. This way we encode all the information about the partition function of a thermal state into two states with comparatively simple structure.

\begin{example}
\label{exa:StabAndProdPsiGee}\emph{Here  we consider examples of states $\ket{\alpha}$, which encode the interaction strengths of the Hamiltonian function.} For the models we consider, these are product states $\ket{\alpha}=\bigotimes_{e\in E}\left|\alpha_{e}\right\rangle $, which are derived immediately from the respective Hamiltonian functions given in section (\ref{sec:Background-on-spin}).
\end{example}
\begin{enumerate}
\item For the $q$-state Potts model the state $\ket{\alpha}$ is derived from the Hamiltonian function~\eqref{eq:HamPottsClock}, with the function $J$ given by $J_{\mbox{\tiny{Potts}}}(\Theta_{ij}):=-\varepsilon\delta(\Theta_{ij})$. This Hamiltonian function is characterized by two-body interactions, whose strengths are encoded into states $\left|\alpha_{e}\right\rangle $ which take the form 
\[
\left|\alpha_{e}\right\rangle =\left|\alpha\right\rangle _{\mbox{\tiny{Potts}}}
=e^{\beta\varepsilon}\left|0\right\rangle +\sum_{j=1}^{q-1}\left|j\right\rangle .
\]

\item For the $q$-state clock model the state $\ket{\alpha}$ is derived from the Hamiltonian function~\eqref{eq:HamPottsClock}, with the function $J$ given by $J_{\mbox{\tiny{clock}}}(\Theta_{ij}):=-\varepsilon\cos(\Theta_{ij})$. The individual two-body interaction strengths are thus encoded into states
\[
\left|\alpha_{e}\right\rangle =\left|\alpha\right\rangle _{\mbox{\tiny{clock}}}
=\sum_{j=0}^{q-1}e^{\beta\varepsilon\cos\left(2\pi j/q\right)}\left|j\right\rangle .
\]

\item As a special case, for $q=2$ we obtain, in an analogous fashion, the states $\ket{\alpha}$ and $\ket{\alpha_{e}}$ for the Ising model
\[
\left|\alpha_{e}\right\rangle =\left|\alpha\right\rangle _{\mbox{\tiny{Ising}}}
=\left|0\right\rangle +e^{-\beta J}\left|1\right\rangle .
\]

\end{enumerate}
\emph{In the following part we look at examples of states $\ket{\psi_{G}}$, which encode the interaction patterns of associated Hamiltonian functions}, thereby investigating special cases of graphs and their corresponding stabilizer states.

\begin{enumerate}
\item \emph{Tree Graphs.} Here we consider models whose interaction patterns are characterized by tree graphs, i.e., graphs containing no loops. The statement that ``$n$ columns $\left\{ c_{i}\right\} _{i=1,...,n}$ of the incidence matrix $B^{\sigma}$ of a graph $G$ are linearly dependent'' means that there is a non-trivial linear combination such that $\sum_{i=1}^{n}\lambda_{i}c_{i}=0$. Hence there is at least one vector that equals the negative sum of the remaining ones, say, $c_{1}=-\lambda_{1}^{-1}\sum_{i=2}^{n}\lambda_{i}c_{i}$. Since the columns describe the start and end points of the edges, this means that the graph contains a loop. In turn, loop-less graphs (= tree graphs) have an incidence matrix with $N=\left|E\right|$ linearly independent columns and hence $N$ linearly independent rows. This means that the rows span the entire space  $\mathbb{Z}_{q}^{N}$ $\left(=C_{G}\left(q\right)\right)$ and hence 
\begin{equation} 
\left|\psi_{G}\right\rangle 
=\sum_{v\in\mathbb{Z}_{q}^{N}}\left|v\right\rangle
\propto\left(\sum_{j=0}^{q-1}\left|j\right\rangle \right)^{\otimes N}.
\end{equation}
In conclusion, we observe that the states derived from tree-graphs are product states.
\item \emph{A cycle.} Here we consider models whose interaction patters are cycles, i.e.\ a closed loop. If the graph is a closed chain,  the incidence matrix looks (besides reordering of the edges) like this
\[
B^{\sigma}=\left(\begin{array}{ccccc}
-1 & 1 & 0 &  & 0\\
0 & -1 & 1 & \cdots & 0\\
0 & 0 & -1 &  & 0\\
& \vdots &  & \ddots & 1\\
1 & 0 & 0 & \cdots & -1\end{array}\right).
\]
We see that a vector $v$ that is perpendicular to all rows has the property $v_{i}=v_{j}$ for all $i$ and $j$, and hence  $C_{G}\left(q\right)^{\perp}=\mathrm{span}\left\{ \left(1,1,1,1,...,1\right)^T\right\} $. We hence choose
\[
\left\{ Z^{\otimes N},X^{\left(n\right)}\left(X^{-1}\right)^{\left(n+1\right)}|n=1,...,N-1\right\}
\]
as  generating set of the stabilizer. We can verify that the state  $\left|\psi_{G}\right\rangle =\sum_{j=0}^{N-1}\left(\left|j_{x} \right\rangle ^{\otimes N}\right)$, where $\left|j_{x}\right\rangle $ is an eigenstate of the $X$-operator, is an eigenstate of the generator of the stabilizer and hence the stabilizer itself. This state is invariant under reordering of the edges and hence the proof is independent of the choice of $B^{\sigma}$ that was chosen in the beginning. Thus, the states derived from graphs that are closed chains are (generalized) Greenberger-Horne-Zeilinger states (GHZ states.)  In particular, for $q=2$ one obtains the state  $|+\rangle^{\otimes N} + |-\rangle^{\otimes N}$ (where $|+\rangle$ and  $|-\rangle$ are the eigenstates of the Pauli matrix $\sigma_x$).
\item \emph{The Kitaev model.} The Kitaev model of topologically protected quantum states is defined as follows. On each edge of a toric lattice with checkerboard structure we place one qubit, the edge qubit. The toric code state (actually a subspace) is the common eigenstate of a set of operators that are constructed using the neighborhood relations of the toric lattice. More precisely, for each but one of the smallest possible loops $L_{i}$ (the \emph{plaquettes}) in the lattice, we define one operator \[ B_{i}:=\prod_{\left(a,b\right)\in L_{i}}Z^{\left(a,b\right)}.\]  We leave out one because it would not be independent from the others. Similarly, each vertex $a$ (there is no qubit in the vertices) has a neighborhood $N_{a}$ of adjacent edges, forming a star. On the qubits of each but one of these stars we define the operators 
\[
A_{a}:=\prod_{b\in N_{a}}X^{\left(a,b\right)}.
\]
One has to be left out because it is not independent of the others, as in the case of the plaquettes. All these operators mutually commute, because in each loop a vertex has either zero or two nearest neighbors. Hence, these operators generate a stabilizer, whose fixed point is the toric code state. We notice that this stabilizer consists of $2^{2N-2}$ independent operators defined on a $2N$-site quantum system and hence the stabilized object is not a single state but a subspace of dimension $4$. We remark that the connection between the 2D Ising model and planar (toric) code states was first proven and utilized in Ref.~\cite{BR06}.\\ In view of the huge variety of classical spin models and their interaction graphs, we want to point out that this state can be defined more abstractly and more closely related to the $B^{\sigma}$-matrix construction used in the other examples, as shown in the following: We assume $q=2$ and consider an arbitrary graph $G^{\sigma}=\left(V,E\right)$ with the essential property to contain $N-1=\left|V\right|-1$ independent loops $\left\{ L_{i}\right\} _{i=1}^{N-1}$. The loops now naturally define a specific neighborhood $N_{a}$ of each vertex $a$, namely the union of the sets of nearest neighbors of $a$ in each loop  $N_{a}=\bigcup_{L_{i}}\left\{ b\in V;\left(a,b\right)\in L_{i}\right\} $. With the loops and neighborhoods specified, we define as above the  operators
\[
A_{a}:=\prod_{b\in N_{a}}X^{\left(a,b\right)};\quad B_{i}
:=\prod_{\left(a,b\right)\in L_{i}}Z^{\left(a,b\right)}.
\]
All of these operators mutually commute, because in each loop a vertex has either zero or two nearest neighbors. There are $N-1$ independent operators $A_{a}$ and $N-1$ independent operators $B_{i}$, which can be seen as follows. Considering the operators $B_{i}$, the statement $\prod_{i\in I}B_{i}=\one$ with a set of loops $I$ implies that $I$ contains dependent loops, which is impossible for $\left|I\right|<N$. Similarly, considering the operators $A_{a}$, for each set of vertices $V'$ the identity $\prod_{a\in V'}A_{a}=\one$ means that the sets $N_{a}=\left\{ \left(a,b\right),\quad b\in N_{a}\right\} $ (when considered together) contain each edge twice. This is impossible if $\left|V'\right|<N$ by construction of the interaction graph. On the other hand $\prod_{a\in V}A_{a}=\one$ and $\prod_{i}B_{i}=\one$ by similar arguments. As a special instance we recover the example given above if $G$ is the periodic two-dimensional lattice, where the common fixed point of the operators $A_{a}$ and $B_{i}$ defines the toric code state~\cite{Ki03}, as introduced in the context of topological quantum computation.
\end{enumerate}

We will use yet another generalization of this construction procedure in the subsequent section, in order to access correlation functions and partition functions of systems with external magnetic fields.

\subsection{External fields and correlation functions}

The encoding scheme discussed in the last section is neither suited to evaluate correlation functions nor partition functions of systems with external fields, like those described by 
\[
H\left(\left\{ s_{i}\right\} \right)
=\sum_{e\in E}h_{e}\bigl(\bigl|s_{v_{e}^{+}}-s_{v_{e}^{-}}\bigr|_{q}\bigr)
+\sum_{v\in V}b_{v}\bigl(s_{v}\bigr).
\]
To overcome this limitation we (have to) use a different encoding scheme. Instead of the state $\left|\psi_{G}\right\rangle $ we will now use the state 
\[
\left|\varphi_{G}\right\rangle 
:=\sum_{\mathbf{s}\in\mathbb{Z}_{q}^{N}}\left|\mathbf{s}\right\rangle 
\bigl|\left(B^{\sigma}\right)^T\mathbf{s}\bigr\rangle
\]
to encode the interaction pattern, where $B^{\sigma}$ is again the incidence matrix of the interaction graph $G$.

\begin{lem}
\label{lem:stabilizer-state2}The state $\left|\varphi_{G}\right\rangle $ is a stabilizer state. Its stabilizer is generated by the $N$ operators
\begin{eqnarray*}
K_{a} & = & X^{\left(a\right)}\prod_{e: \exists b\in V s.th. (a,b)=e\in E}
\left(X^{\left(e\right)}\right)^{\sigma}\\
K_{e} & = & Z^{\left(e\right)}\left(Z^{\left(a\right)}\right)^{-\sigma}
\left(Z^{\left(b\right)}\right)^{\sigma},
\end{eqnarray*}
for every $a\in V$ and for every $e=\left(a,b\right)\in E$, where $\sigma$ is either $+1$ or $-1$, depending on the orientation of the edge ($\sigma:=B_{e,a}^{\sigma}=-B_{e,b}^{\sigma}$). In our notation, an upper index in brackets denotes the qudit acted on by the operator. 
\end{lem}
\emph{Proof.} We have to show that these operators have $\left|\varphi_{G}\right\rangle $ as a fixed point and, since the number of operators defined this way is $N=\left|V\right|+\left|E\right|$ and hence equals the number of qudits in the quantum system in state $\left|\varphi_{G}\right\rangle $, we have to show that they are independent. Under this condition they generate the stabilizer of a single state. To see the stabilizing property of the operators $K_a$, we compute
\begin{multline*}
X^{\left(a\right)}\prod_{e: \exists b\in V s.th. (a,b)=e\in E}
\left(X^{\left(e\right)}\right)^{\sigma}\left|\mathbf{s}\right\rangle \bigl|
\left(B^{\sigma}\right)^T\mathbf{s}\bigr\rangle\\
=\left|\mathbf{s}'\right\rangle \bigl|\left(B^{\sigma}\right)^T\mathbf{s}'\bigr\rangle
\end{multline*}
with $\mathbf{s}'=\left(s_{0},...,s_{a}+1\mbox{ mod }q,...,s_{\left|V\right|}\right)$, because $X^{\left(a\right)}\left|\mathbf{s}\right\rangle =\left|\mathbf{s}'\right\rangle $ and 
\begin{multline*}
\prod_{b:\left(a,b\right)
=e\in E}\left(X^{\left(e\right)}\right)^{\sigma}\bigl|
\left(B^{\sigma}\right)^T\mathbf{s}\bigr\rangle\\
=\left|\left(c_{a}+\left(B^{\sigma}\right)^T\mathbf{s}\right)
\mbox{ mod }q\right\rangle =\left|\left(B^{\sigma}\right)^T\mathbf{s}'\right\rangle ,
\end{multline*}
where $c_{a}$ is the $a$th column of $B^{\sigma}$. Likewise,   $K_{e}\left|\varphi_{G}\right\rangle =\left|\varphi_{G}\right\rangle $, because 
\begin{multline*}
\left(Z^{\left(a\right)}\right)^{-\sigma}
\left(Z^{\left(b\right)}\right)^{\sigma}\left|\mathbf{s}\right\rangle \\
=\left(Z^{\left(a\right)}\right)^{-B_{e,a}^{\sigma}}
\left(Z^{\left(b\right)}\right)^{-B_{e,b}^{\sigma}}
\left|\mathbf{s}\right\rangle \\
=\exp\left\{ -2\pi i\left(B_{e,a}^{\sigma}s_{a}
+B_{e,b}^{\sigma}s_{b}\right)/q\right\} \bigl|
\left(B^{\sigma}\right)^T\mathbf{s}\bigr\rangle,
\end{multline*}
and
\begin{multline*}
Z^{\left(e\right)}\bigl|\left(B^{\sigma}\right)^T\mathbf{s}\bigr\rangle\\
=\exp\left\{ 2\pi i\left(B_{e,a}^{\sigma}s_{a}
+B_{e,b}^{\sigma}s_{b}\right)/q\right\} \bigl|\left(B^{\sigma}\right)^T\mathbf{s}\bigr\rangle,
\end{multline*}
so the phases cancel.

The set of $2n$ operators that were just defined are mapped, by the isomorphism $F_{q}^{2n}$, to a set of $2n$ vectors that can be arranged in the following matrix
\[
\left(\begin{array}{cc}
\one_{\left|V\right|} & 0\\
\left(B^{\sigma}\right)^T & 0\\
0 & -B^{\sigma}\\
0 & \one_{\left|E\right|}\end{array}\right).
\]
This matrix has full rank, considering the $\one s$. Hence the operators generate the full stabilizer. $\hfill\square$

In the case $q=2$ we recover a true graph state by an application of Hadamard transformation on the edge qubits. For an illustrative example, see Fig.~\ref{fig:decorated}. As before, the classical spin sites correspond to vertices in the interaction graph G of the classical model. The interacting pairs of sites are then mapped to a quantum site, one for each edge (the \emph{edge qudits}). What is different from the original scheme is that the individual classical spin sites --acted on by local fields -- are mapped to quantum sites as well, one for each vertex (the \emph{vertex qudits}). The resulting graph is called a \emph{decorated graph}. The resulting many body quantum states are again, by construction, stabilizer states.

\begin{figure*}
\includegraphics[width=1.2\columnwidth]{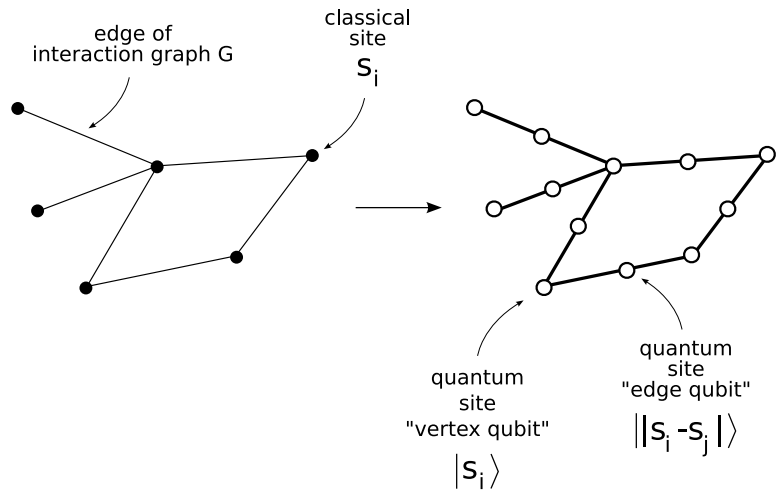}

\caption{\emph{The extended construction principle.} This figure shows an example of the extended encoding of a classical interaction pattern into a stabilizer state. Thin graph on the left: the classical interaction graph G; thick graph on the right: the derived graph relating quantum sites in a stabilizer state. The classical spin sites correspond to vertices in a graph G. The interacting pairs of sites are mapped to a quantum site, one for each edge (\emph{edge qudits}). The individual classical spin sites on which the local fields act are then, too, mapped to quantum sites, one for each vertex (\emph{vertex qudits}) -- this is different from the original scheme. The resulting graph is called a \emph{decorated graph}. The quantum sites are, by construction, in a stabilizer state.}

\label{fig:decorated}
\end{figure*}

\subsubsection{Thermal quantities}

We now come to the central result of this section. By means of the state $\left|\varphi_{G}\right\rangle $ and appropriately chosen product states, we can compute the partition function of systems with local external fields as well as n-point functions.

\begin{thm}
\label{thm:partfun_ext}The partition function  $Z_{G}\left(\left\{ h_{e},b_{v}\right\},\beta\right)$ of a classical  spin system at inverse temperature $\beta$, defined on the graph $G=\left(V,E\right)$ by the Hamiltonian function  \textup{$H\left(\left\{ s_{i}\right\} \right) =\sum_{e\in E}h_{e}\bigl(\bigl|s_{v_{e}^{+}}-s_{v_{e}^{-}} \bigr|_{q}\bigr)+\sum_{v\in V}b_{v}\bigl(s_{v}\bigr)$}, can be written as the overlap of a stabilizer state and a product state
\begin{eqnarray*}
Z_{G}\left(\left\{ h_{e},b_{v}\right\} ,\beta\right) 
& = & (\bigotimes_{v\in V} \bra{\alpha'_{v}}\bigotimes_{e\in E}
\bra{\alpha_{e}})\ket{\varphi_{G}},
\end{eqnarray*}
where 
\begin{eqnarray*}
\ket{\alpha_{e}}
& = & \sum_{j=0}^{q-1}e^{-\beta h_{e}\left(j\right)}\ket{j}\\
\ket{\alpha'_{v}}
& = & \sum_{j=0}^{q-1}e^{-\beta b_{v}\left(j\right)}\ket{j}.
\end{eqnarray*}

\end{thm}
\begin{proof}
The state $\left|\varphi_{G}\right\rangle $ is a stabilizer state according to lemma (\ref{lem:stabilizer-state2}), and we compute, with an arbitrarily chosen orientation $\sigma$ of the graph $G$,
\begin{multline*}
(\bigotimes_{v\in V} \bra{\alpha'_{v}} \bigotimes_{e\in E}
\bra{\alpha_{e}}) \ket{\varphi_{G}}\\
=(\bigotimes_{v\in V}\bra{\alpha'_{v}}\bigotimes_{e\in E}
\bra{\alpha_{e}})
\sum_{\mathbf{s}\in\mathbb{Z}_{q}^{N}}\ket{\mathbf{s}}
\ket{\left(B^{\sigma}\right)^T\mathbf{s}}\\
=\sum_{\mathbf{s}\in\mathbb{Z}_{q}^{N}} \prod_{v \in V} e^{-\beta b_{v}\left(s_{v}\right)}
\prod_{e\in E}e^{-\beta h_{e}(\bigl|s_{v_{e}^{+}}-s_{v_{e}^{-}}\bigr|_{q})}\\
=\sum_{\mathbf{s}\in\mathbb{Z}_{q}^{N}} e^{-\beta H\left(\left\{ s_{i}\right\} \right)}.
\end{multline*}
\end{proof}

Likewise, we can write down a theorem for the n-point correlation functions.

\begin{thm}
The n-point correlation functions $\left\langle s_{i_{1}}, s_{i_{2}}, ..., s_{i_{n}} \right\rangle_{\beta}$ of a classical spin system at inverse temperature $\beta$, defined on the graph $G^{\sigma}=\left(V,E\right)$ by the Hamiltonian function  \textup{$H\left(\left\{ s_{i}\right\} \right) =\sum_{e\in E}h_{e}\bigl(\bigl|s_{v_{e}^{+}}-s_{v_{e}^{-}}\bigr|_{q}\bigr) +\sum_{v\in V}b_{v}\bigl(s_{v}\bigr)$}, can be written as an overlap of a stabilizer state and a product state (up to a factor of $Z$, which is the partition function of the classical spin system). More precisely,
\begin{eqnarray*}
\left\langle s_{i_{1}},s_{i_{2}},...,s_{i_{n}}\right\rangle_{\beta} & =Z^{-1} &
(\bigotimes_{v\in V}\bra{\alpha'_{v}\left(i_{1},...,i_{n}\right)}
\bigotimes_{e\in E}\bra{\alpha_{e}})\ket{|\varphi_{G}},
\end{eqnarray*}
where
\begin{eqnarray*}
\ket{\alpha_{e}}&
= & \sum_{j=0}^{q-1}e^{-\beta h_{e}\left(j\right)}\ket{j}\\
\ket{\alpha'_{v}\left(i_{1},...,i_{n}\right)}&
= &\sum_{j=0}^{q-1}\cos\left(2\pi j/q\right)^{m_{\nu}}
e^{-\beta b_{v}\left(j\right)}\ket{j},
\end{eqnarray*}
and $m_{\nu}$ is the number of occurrences of $\nu$ in the n-tuple $\left(i_{1},...,i_{n}\right)$.
\end{thm}
\begin{proof}
The state $\ket{\varphi_{G}}$ is a stabilizer state according to lemma (\ref{lem:stabilizer-state2}), and we compute
\begin{multline*}
(\bigotimes_{v\in V}\left\langle \alpha'_{v}
\left(i_{1},...,i_{n}\right)\right|\bigotimes_{e\in E}
\bra{\alpha_{e}})\ket{\varphi_{G}}\\
=(\bigotimes_{v\in V}\bra{\alpha'_{v}\left(i_{1},...,i_{n}\right)}
\bigotimes_{e\in E}\bra{\alpha_{e}})
\sum_{\mathbf{s}\in\mathbb{Z}_{q}^{N}}\ket{\mathbf{s}} 
\ket{\left(B^{\sigma}\right)^T\mathbf{s}}\\
=\sum_{\mathbf{s}\in\mathbb{Z}_{q}^{N}}\prod_{v\in V}
\cos\left(2\pi s_{\nu}/q\right)^{m_{\nu}}e^{-\beta b_{v}\left(s_{v}\right)}
\prod_{e\in E}e^{-\beta h_{e}(\bigl|s_{v_{e}^{+}}-s_{v_{e}^{-}}\bigr|_{q})}\\
=\sum_{\mathbf{s}\in\mathbb{Z}_{q}^{N}}\cos\left(\Theta_{i_{1}}\right)
\cos\left(\Theta_{i_{2}}\right)...\cos\left(\Theta_{i_{n}}\right)
e^{-\beta H\left(\left\{ s_{i}\right\} \right)}.
\end{multline*}
where $\Theta_{i}=2 \pi s_i/q$. We compare this to the definition of the n-point correlation function given in section \ref{sec:Background-on-spin}. This concludes the theorem.
\end{proof}

\section{\label{sec:Extending-the-formalism}Extending the formalism}

\subsection{The most general framework}

So far we have used product states of single edge-qudit sites, namely the states $\left|\alpha\right\rangle =\bigotimes_{e\in E} \left|\alpha_{e}\right\rangle $, in the overlap with the states $\left|\varphi_{G}\right\rangle $ and $\left|\psi_{G}\right\rangle $, to calculate partition functions and correlation functions. Allowing for tensor products of entangled states, $\left|\alpha\right\rangle =\bigotimes_{\varepsilon\subset E} \left|\alpha_{\varepsilon}\right\rangle $, where the $\varepsilon$ are subsets of $E$ with few elements, extends the set of possible encodings of classical spin systems. This is the content of this section.

One shortcoming of the encoding of the interaction graph into the states $\left|\varphi_{G}\right\rangle $ and $\left|\psi_{G}\right\rangle $ is the inability of the interaction to distinguish between classical spin states that have the same relative state  $\left|s_{i}-s_{j}\right|_q=\left|s'_{i}-s'_{j}\right|_q$ but have different values $s_{i}\neq s'_{i}$, $s_{j}\neq s'_{j}$. This inability stems from the fact that an attempt to encode pairs of neighboring states $\left(s_{i},s_{j}\right)$ into one edge qudit $\left(s_{i},s_{j}\right)\mapsto\left|e_{ij}\right\rangle $ via the $B$-matrix formalism does not lead to a stabilizer state and hence fails, \emph{if $\left|e_{ij}\right\rangle $ takes more states than each of the sites $s_{i}$ or $s_{j}$.} One way out of this dilemma is to encode the pairs of neighboring spin sites in the graph of the classical model into more than one qubit, while extending the overlap state $\left|\alpha\right\rangle $ to states beyond  product states. Although these states are not product states anymore, we can still interpret them as product states of composite particles, extending over few sites as we restrict ourselves to subsets $\varepsilon$ of $E$ with few elements. The entangled states moreover include neighboring sites only, which adds to the picture of composite sites (quasi-local states).

We describe this generalization now and investigate the relationship to the more specialized cases. Under certain assumptions concerning the classical Hamiltonian function, a formal mapping from the most general case to the more specialized ones is possible. Taking this step, i.e. performing this formal transformation, gives us a mathematical picture which is often much more enlightening than the original one.

\subsection{Encoding $m$-body interactions, each site appearing in maximally $n$ terms of the Hamiltonian function.}

The most general case to consider is the one where we

\begin{itemize}
\item allow each classical spin to appear in as many as $n$ terms of the Hamiltonian function
\item allow each site to interact with $m-1$ others (Hamiltonian function with $m$-body terms)
\item allow all configurations of the $m$ interacting spins in each term to be differentiated energetically.
\end{itemize}
Note however, that a simulation of thermal states of these systems on a classical computer scale unfavorably  in $m$ and $n$, as we will see in section \ref{sub:Simulations-on-classical}.

The first point in the list is addressed in the following way. Since each site is allowed to take part in $n$ interactions, we need $n$ instances of it in the stabilizer state. Of course, all instances of the local quantum systems have to be in the same state when measured. Hence we map each site $e_{i}$ to an $n$-body GHZ state: $e_{i} \mapsto \sum \ket{s_{i,1}s_{i,2}...s_{i,n}}$. To address the latter two points of this list, we consider the following. To create a quantum state $\left| \gamma_{G} \right\rangle $ that enables us to differentiate energetically between all possible spin configurations of an $m$- body interaction, we map each site $e_i$ taking part in the interaction to a single quantum spin state $e_{i}\mapsto\ket{s_{i}}_{e_{i}}$. The corresponding state $\ket{\alpha}$, which is used for the contraction that yields the partition function and which in the preceding sections used to be a product state, consequently has to be an entangled state in this picture. On the sites $\left\{ e_{i}\right\} _{i=1}^{m}$ taking part in one $m$-body interaction, the state $\ket{\alpha}$ takes the form
\[
\ket{\alpha}=\sum_{\left(s_{1},s_{2},...,s_{m}\right)}
e^{-\beta h\left(s_{1},s_{2},...,s_{m}\right)}\ket{s_{1}s_{2}...s_{m}}.
\]
Note however, that a simulation of thermal states of these systems on a classical computer scales unfavorably in $m$ and $n$, as we will see in section \ref{sub:Simulations-on-classical}. For further details of this encoding, let us now have a look at examples.

\subsubsection{Encoding $2$-body interactions, each site appearing in maximally $n$ terms of the Hamiltonian function: Edge models}

A special case of the discussed generalization is the one where we (as in the preceding sections) stick to Hamiltonian functions with $2$-body terms where each site is involved in $n$ interactions.  This kind of Hamiltonian function plays a role in higher dimensional lattices  and spin glasses, for example. For their treatment we propose, in the following, a way to discriminate  the classical spin configurations  beyond resolving relative states (as we did before).

To create a quantum state $\left|\gamma_{G}\right\rangle $ that enables us to differentiate energetically between all possible spin configurations, we proceed as follows. We identify two qudits with each edge $e=\left(ij\right)\in E$ of the graph $G$ and provide them with a product basis  $\left\{ \left|s_{i}\right\rangle _{e_{i}}\left|s_{j}\right\rangle _{e_{j}} | s_i, s_j = 0 ... q-1 \right\}$, where $e_{i}$ is one of the edge qudits and $e_{j}$ is the other one. These qudits will be called the \emph{edge qudits} corresponding to the edge. We map states of the classical spin sites to quantum state of the whole quantum many body system of edge qudits via
\[
\mathbb{Z}_{q}^{N}\ni\mathbf{s}=\left(s_{0},...,s_{N}\right)
\mapsto\bigotimes_{e=\left(ij\right)\in E}
\left|s_{i}\right\rangle_{e_{i}}\left|s_{j}\right\rangle _{e_{j}}.
\]
This way, we attach GHZ-states to the vertices, with the number of particles equaling the number of incident edges. A graphical representation of this encoding is given in Fig.~\ref{fig:ghzscheme} (a). Note that each classical spin is mapped to as many edge qudits as there are edges attached to the classical spin vertex.%
\begin{figure*}
\includegraphics[width=1.2\columnwidth]{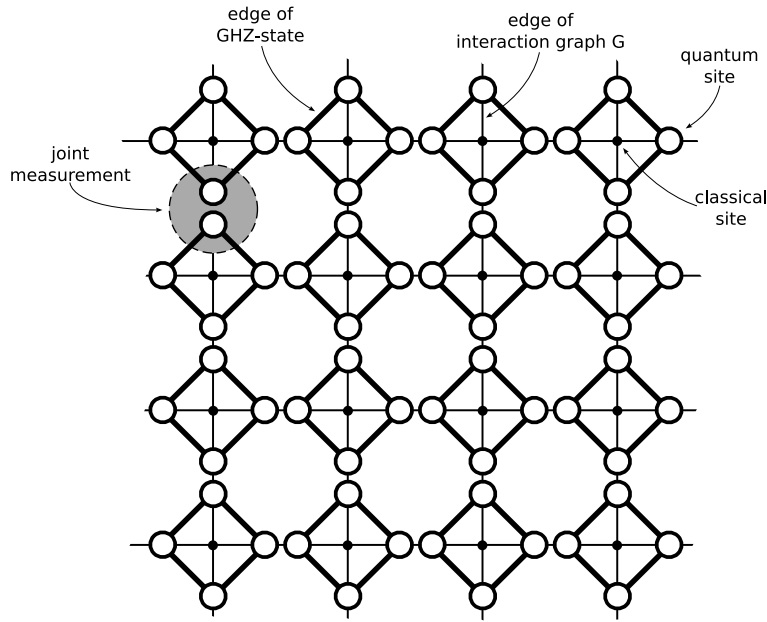}

\caption{\emph{Alternative encoding schemes I: Edge models. (The GHZ scheme.)}  This figure shows an example of an encoding of a classical interaction pattern into a product of GHZ states. The classical interaction graph, a square lattice in this example, is given by the underlying thin grid, the vertices symbolizing classical spin sites and the edges symbolizing their interactions. Each edge holds a pair of edge qubits, as indicated by the dots. The edge-qubits that belong to the same classical spin site are connected by thick lines, indicating that they form a GHZ state. The circle with dashed circumference indicates one pair of qubits $\left|s_{i}\right\rangle \left|s_{j}\right\rangle $ contracted with one state $\left|\alpha_{ij}\right\rangle $ in the Hilbert space of the pair of edge qubits. }

\label{fig:ghzscheme}
\end{figure*}

\begin{lem}
\label{lem:gammaGisGHZ}The superposition of the quantum states belonging to all possible classical states
\[
\left|\gamma_{G}\right\rangle 
:=\sum_{\mathbf{s}\in\mathbb{Z}_{q}^{N}}\bigotimes_{e=\left(ij\right)\in E}
\left|s_{i}\right\rangle _{e_{i}}\left|s_{j}\right\rangle _{e_{j}}
\]
is a product of GHZ-states and hence a stabilizer state.
\end{lem}
\emph{Proof.} A reordering of the sites groups all edge qubits belonging to the same spin site $i$ 
\[
\bigotimes_{e=\left(ij\right)\in E}\left|s_{i}\right\rangle _{e_{i}}
\left|s_{j}\right\rangle _{e_{j}}=\bigotimes_{i\in V}
\bigotimes_{e=\left(ij\right)}\left|s_{i}\right\rangle _{e_{i}}
\]
and writing it this way we see that the state $\left|\gamma_{G}\right\rangle $ has the structure
\[
\left|\gamma_{G}\right\rangle 
=\sum_{\mathbf{s}\in\mathbb{Z}_{q}^{N}}\bigotimes_{i\in V}
\bigotimes_{e=\left(ij\right)}\left|s_{i}\right\rangle _{e_{i}}
\overset{\mbox{\tiny{reordering}}}{\mapsto}\bigotimes_{i\in V}\sum_{s_i\in\mathbb{Z}_{q}}
\bigotimes_{e=\left(ij\right)}\left|s_{i}\right\rangle _{e_{i}},
\]
where $\sum_{s_i\in\mathbb{Z}_{q}}  \bigotimes_{e=\left(ij\right)}\left|s_{i}\right\rangle _{e_{i}}$ is a GHZ state.$\hfill\square$

The overlap to evaluate the partition function or correlation functions has now to be performed with one state per edge qubit \emph{pair} $\braket{ \alpha_{e}}{s_{i}s_{j}}$. Since this overlapping state $\ket{\alpha}$ allows us to adapt the energies $h_{ij}$ to each individual spin, the possibility of evaluating partition functions with local energy terms as well as correlation functions is immediately given.

To avoid the necessity of encoding the correlation function directly into the states $\ket{\alpha_{e}}$, we add one more quantum site to the GHZ state. This enables us to measure the state of the classical site directly. Keep in mind that this is technically not necessary, because the state of the site is directly accessible already without the extension.

\begin{thm}
\label{thm:partfun-ghz}The partition function  $Z_{G}\left(\left\{ h_{e},b_{v}\right\},\beta\right)$ of a classical spin system at inverse temperature $\beta$, defined on the graph $G=\left(V,E\right)$ by the Hamiltonian function  \textup{$H\left(\left\{ s_{i}\right\} \right) =\sum_{\left(ij\right)\in E}h_{\left(ij\right)}\bigl(s_{i},s_{j}\bigr)$}, can be written as the overlap of a stabilizer state and a product state (over edge qudit pairs) 
\begin{eqnarray*}
Z_{G}\left(\left\{ h_{ij},b_{v}\right\} ,\beta\right)
& = & (\bigotimes_{\left(ij\right)\in E}
\bra{\alpha_{\left(ij\right)}})\ket{\gamma_{G}},
\end{eqnarray*}
where
\begin{eqnarray*}
\ket{\alpha_{\left(ij\right)}}
& = & \sum_{s_{i},s_{j}=1}^{q}e^{-\beta h_{ij}\left(s_{e_{i}},s_{e_{j}}\right)}
\ket{s_{i}}_{e_{i}}\ket{s_{j}}_{e_{j}}
\end{eqnarray*}

\end{thm}
\begin{proof}
The state $\left|\gamma_{G}\right\rangle $ is a product of GHZ states and hence a stabilizer state according to lemma \ref{lem:gammaGisGHZ}, and we compute, with an arbitrarily chosen orientation $\sigma$ of the graph $G$,
\begin{multline*}
(\bigotimes_{\left(ij\right)\in E} \bra{\alpha_{\left(ij\right)}})
\ket{\gamma_{G}}\\
=(\bigotimes_{\left(ij\right)\in E}\bra{\alpha_{\left(ij\right)}})
\sum_{\mathbf{s}\in\mathbb{Z}_{q}^{N}}\bigotimes_{e=\left(ij\right)\in E}
\ket{s_{i}}_{e_{i}} \bra{s_{j}}_{e_{j}}\\
=\sum_{\mathbf{s}\in\mathbb{Z}_{q}^{N}}
\prod_{\left(ij\right)\in E}e^{-\beta h_{ij}\left(s_{e_{i}},s_{e_{j}}\right)}\\
=\sum_{\mathbf{s}\in\mathbb{Z}_{q}^{N}}
e^{-\beta H\left(\left\{ s_{i}\right\} \right)}.
\end{multline*}
\end{proof}

\subsubsection{Encoding $4$-body interactions, each site appearing in maximally $2$ terms of the Hamiltonian function: Vertex models}

An important class of models are the vertex models. These models also fit into our framework, as will be shown now. A prominent example of a vertex model stems from a $2D$ regular lattice where each classical site interacts with two groups of three neighboring particles (each individually) (see Fig.~\ref{fig:vertexscheme}). Hence we have the situation where $4$-body interactions take place, each site appearing in maximally $2$ terms of the Hamiltonian function. Consequently, we encode each classical spin site into a $2$-body GHZ-state (Bell state), and entangle the quartets of sites, corresponding to the interactions, in the state $\ket{\alpha}$.

This setting yields a vertex model in two dimensions, where the projections (of the subsets of the GHZ states taking place at each vertex of the vertex model) are determined by the set of states $\ket{\alpha}$. Similar models in higher dimensions can be obtained easily in an analogous fashion.
\begin{figure*}
\includegraphics[width=1.2\columnwidth]{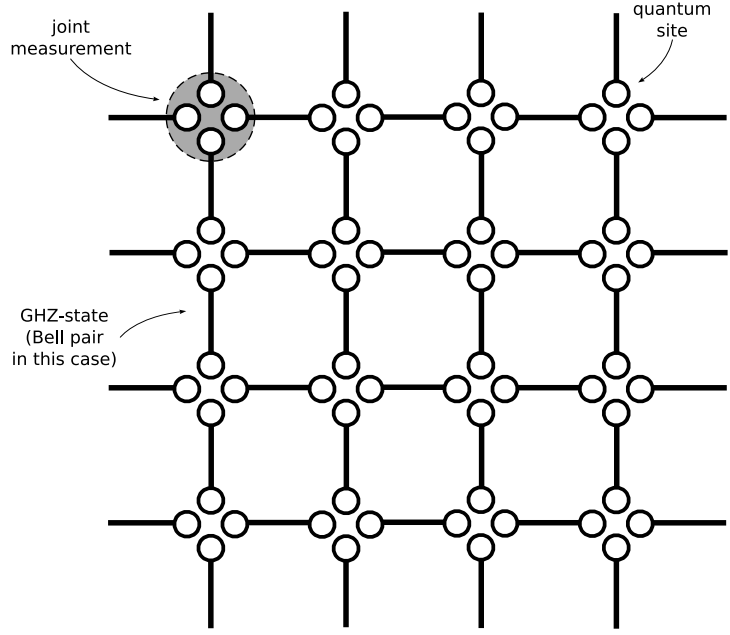}

\caption{\emph{Alternative encoding schemes II: Vertex models.} This figure shows an example of an encoding of a classical interaction pattern into a vertex model, where each thick line represents one Bell pair. In this example, each classical spin site enters in two four-site classical interactions. Accordingly, four edge qubits (circle with dashed circumference) form the smallest subsystem of the Hilbert space used for the overlap states $\left|\alpha_{\varepsilon}\right\rangle $.}

\label{fig:vertexscheme}
\end{figure*}

\subsection{Relations between the encoding schemes}

An interesting question is how the generalized models that were just described relate to the encoding scheme encompassing the states  $\left|\varphi_{G}\right\rangle $ and $\left|\psi_{G}\right\rangle $. We want to discuss this now and furthermore give additional relations between the states  $\left|\varphi_{G}\right\rangle $ and $\left|\psi_{G}\right\rangle $, adding to what was presented in the preceding sections.

\subsubsection{Relations between $\left|\varphi_{G}
\right\rangle $ and $\left|\psi_{H}\right\rangle $}

There are instances where for different graphs $G$ and $H$ the cutspaces of $\left|\varphi_{G}\right\rangle $ and $\left|\psi_{H}\right\rangle $ are closely related. Two examples will now be demonstrated and give us some more insight into the internals of the construction.

The first way to look at the construction of the cutspace of  $\left|\varphi_{G}\right\rangle $ is to modify the graph $G$ by changing the mapping of $G$ to the quantum spin sites. We remember that (in the case of the construction of $\left|\psi_{G}\right\rangle $), the method was to map each edge to one quantum spin site. As an alternative, we derive now from $G$ an new graph by placing on each edge one additional vertex. This new graph we call the \emph{decorated graph} $\tilde{G}$, which possesses $N=\left|V\right|+\left|E\right|$ vertices. The crucial point is now to identify the \emph{vertices} in $\tilde{G}$ with the qudits that we chose as a product basis in the definition of the state  $\left|\varphi_{G}\right\rangle $. The original vertices (that appear in $G$ and in $\tilde{G}$) are called \emph{vertex qudits} and the qudits that were added at the edges are called \emph{edge qudits}. The incidence matrix of the decorated graph $\tilde{G}$ is now $\left(\one|B^{\sigma}\right)$ with  $\left|\varphi_{G}\right\rangle =\sum_{\mathbf{s}\in\mathbb{Z}_{q}^{N}}| \left(\one|B^{\sigma}\right)^T\mathbf{s}\bigr\rangle$, where for each summand $\left|\mathbf{s}\right\rangle  \bigl|\left(B^{\sigma}\right)^T\mathbf{s}\bigr\rangle$ the state $\left|\mathbf{s}\right\rangle $ is a state of the vertex qudits and $\bigl|\left(B^{\sigma}\right)^T\mathbf{s}\bigr\rangle$ is a state of the edge qudits. We note that the original method is a restriction of the just proposed mapping of vertices to the edge qudits.

A second way of mapping the graph $G$ to another one that can be used to construct the cutspace of $\left|\varphi_{G}\right\rangle $ is the following. Let us add one vertex to the graph $G$ that is connected to all other vertices. Let us call this vertex $h$ and the new graph $G+h$. The incidence matrix of $G+h$ is
\[
B\left(G+h\right)^T=\left(\begin{array}{cc}
B^{\sigma}\left(G\right)^T & 0\\
\one & \begin{array}{c}
1\\
\vdots\\
1\end{array}\end{array}\right).
\]
The vector of classical spin sites $\mathbf{s}$ has to be extended to include the site $h$, hence we obtain a new vector $\mathbf{s}'= \left(\mathbf{s},s_{h}\right)$. The canonical way to construct  $\left|\psi_{G+h}\right\rangle $ now is 
\begin{eqnarray*}
\left|\psi_{G+h}\right\rangle  
& = & \sum_{\mathbf{s}'}\ket{B\left(G+h\right)^T\mathbf{s}'}\\
& = & \sum_{\mathbf{s},s_{h}}\ket{B^{\sigma}\left(G\right)^T\mathbf{s}}
\ket{\mathbf{s}+\left(s_{h},s_{h},...s_{h}\right)}\\
& = & 2\sum_{\mathbf{s}}
\ket{B^{\sigma}\left(G\right)^T\mathbf{s}}\ket{\mathbf{s}},
\end{eqnarray*}
because for all values of $s_{h}$, the equation
\[
\ket{B^{\sigma}\left(G\right)^T\left(\mathbf{s}
+\left(s_{h},s_{h},...s_{h}\right)\right)}
=\ket{B^{\sigma}\left(G\right)^T\mathbf{s}}
\]
holds. Hence $\left|\psi_{G+h}\right\rangle =2\left|\varphi_{G}\right\rangle $.

A conclusive remark seems appropriate. As has been shown, the stabilizer of the states are derived from the incidence matrix of their interaction graph. In the case of $\left|\psi_{G}\right\rangle $, the span of the rows of $B^{\sigma}$ forms the cutspace directly. Following the arguments in the sections above, the stabilizer of $\left|\varphi_{G}\right\rangle $ is constructed analogously, but from the span of the rows of the matrix  $\left(\one_{\left|V\right|}|B^{\sigma}\right)$ or $B\left(G+h\right)^T$ instead. Although obviously being related, the difference in the construction changes the quantum states qualitatively to a great extend. For instance, when constructing  $\left|\varphi_{G}\right\rangle $ we do not obtain the same state classes  as in the examples (\ref{exa:StabAndProdPsiGee}). Instead of the state  $\left(\sum_{j=0}^{N-1}\left|j_{x}\right\rangle \right)^{\otimes N}$ (in case of a tree graph), or the state  $\sum_{j=0}^{N-1}\left(\left|j_{x}\right\rangle ^{\otimes N}\right)$ (in case of a cycle) or the toric code state, we always obtain states that are locally equivalent to one-dimensional or two-dimensional cluster states, respectively.

Finally, by means of measurements, we are able to obtain the state $\ket{\psi_{G}}$ from the state $\ket{\varphi_{G}}$. By overlapping the vertex qudits of $\left|\varphi_{G}\right\rangle $ with the state $\left(\sum_{j=0}^{q-1}\left|j\right\rangle \right)^{\otimes\left|V\right|}$ we immediately recover $\left|\psi_{G}\right\rangle $. On the one hand, this formally has the meaning of projecting out the dimensions of the state that are stabilized by operators corresponding to the $\one_{\left|V\right|}$-part in the matrix  $\left(\one_{\left|V\right|}|B^{\sigma}\right)$. On the other hand, it has the physical interpretation of setting the
local external fields to zero.

\subsubsection{Going from the general picture to $\left|\varphi_{G}\right\rangle $ and $\left|\psi_{G}\right\rangle $}

The states $\left|\varphi_{G}\right\rangle $ and $\left|\psi_{G}\right\rangle $ encode two-body interactions. Hence the scheme that is a direct super-set of the $\left|\varphi_{G}\right\rangle $ and $\left|\psi_{G}\right\rangle $ encodings is the GHZ scheme. Let this GHZ state be $\ket{GHZ}$.

Contractions with states $\ket{\alpha_{\left(ij\right)}}$ that do not discriminate between quantum states $\ket{s_{i}}_{e_{i}}\ket{s_{j}}_{e_{j}}$ with the same value of $\left|s_{i}-s_{j}\right|_q$ yield directly the appropriate quantum description for $\ket{\psi_{G}}$. To obtain the stabilizer description for this case, each sub-state of $\ket{GHZ}$ consisting of a pair $\ket{s_{i}}_{e_{i}}\ket{s_{j}}_{e_{j}}$ that is measured against a two-qudit state $\ket{\alpha_{\left(ij\right)}}$ is identified with a new single qudit. This qudit in turn corresponds to an edge in the adjacency matrix of the graph $G$ defining the state $\ket{\psi_{G}}$. The vertices $G$ correspond to the GHZ sub-states in the state $\ket{GHZ}$. Hence all the information about the graph $G$ can be recovered from the graphical scheme corresponding to $\ket{GHZ}$. The state $\ket{\alpha'}$ that encodes the interaction strengths is not difficult to find either. Since $\ket{\alpha_{\left(ij\right)}}$ does not discriminate between quantum states $\ket{s_{i}}_{e_{i}}\ket{s_{j}}_{e_{j}}$ with the same value of $\left|s_{i}-s_{j}\right|_q$, we obtain
\[
\ket{\alpha'_{\left(ij\right)}}
=\sum_{s}\left(\sum_{\left|s_{i}-s_{j}\right|_q=s}
e^{-\beta h\left(s_{i},s_{j}\right)}\right)\ket{s}.
\]

Recovering a description of $\ket{GHZ}$ in terms of a state $\left|\varphi_{G}\right\rangle $ can be performed similarly, provided that the Hamiltonian function terms can
be written as
\[
h\left(s_{i},s_{j}\right)
=\bar{h}_{ij}\left(\left|s_{i}-s_{j}\right|_q\right)
+h_{i}\left(s_{i}\right)+h_{j}\left(s_{j}\right).
\]
The GHZ state in the general encoding is a product state of smaller GHZ states
\[
\ket{GHZ}=\bigotimes_{k}\ket{GHZ_{k}}.\]
Each of the states $\ket{GHZ_{k}}$ has to be extended by one site by the mapping
\[
\ket{GHZ_{k}}
=\sum_{s}\bigotimes_{i=1}^{N_{k}}\ket{s}_{i}
\mapsto\sum_{s}\bigotimes_{i=1}^{N_{k}+1}\ket{s}_{i}=:\ket{GHZ_{k}'}.
\]
To obtain the stabilizer description for this case, each sub-state of $\ket{GHZ}$ consisting of a pair $\ket{s_{i}}_{e_{i}}\ket{s_{j}}_{e_{j}}$ that is contracted with a two-qudit state $\ket{\alpha_{\left(ij\right)}}$ is identified with a new single (edge) qudit of the decorated graph corresponding to $\ket{\varphi_{G}}$. This edge-qudit in turn corresponds to an edge in the adjacency matrix $B$ of the graph $G=\left(\one|B\right)$ defining the state $\ket{\varphi_{G}}$. The sub-states that are not measured this way are the ones that were added in the mapping above. These will be used to encode the local fields and hence will be mapped to the vertex qudits of the decorated graph defining the state $\ket{\varphi_{G}}$. The part that is more complicated here than in the case of $\ket{\varphi_{G}}$ is finding the new state $\ket{\alpha'}$ encoding the interaction strengths. To do so, we have to find, for each term of the Hamiltonian function $h\left(s_{i},s_{j}\right)$, a corresponding form $h\left(s_{i},s_{j}\right) =\bar{h}_{ij}\left(\left|s_{i} -s_{j}\right|_q\right)+h_{i}\left(s_{i}\right) +h_{j}\left(s_{j}\right)$. The part $\bar{h}\left(\left|s_{i}-s_{j}\right|_q\right)$ will be encoded in the part of $\ket{\alpha'}$ that is measured against the edge-qudits, e.g.,
\[
\ket{\alpha'_{\left(ij\right)}}
=\sum_{s}\left(\sum_{\left|s_{i}-s_{j}\right|_q=s}
e^{-\beta\bar{h}_{ij}\left(\left|s_{i}-s_{j}\right|_q\right)}\right)\ket{s}.
\]
The local field corresponding to the vertex qudit with states   $\ket{s_{k}}_{N_{k}+1}$, belonging to the extended GHZ sub-state $\ket{GHZ_{k}}$, is found by a summation of all corresponding fields
\[
h_{N_{k}+1}\left(s_{k}\right)=\sum_{j=1}^{N_{k}}h_{j}\left(s_{k}\right),
\]
where $h_{j}\left(s_{j}\right)$ are the new terms of the Hamiltonian function gained from the original terms $h\left(s_{i},s_{j}\right)$, that belong to measurements on sites on the GHZ sub-state $\ket{GHZ_{k}}$.

\section{\label{sec:Applications}Applications}

This section contains applications of the framework given in the sections above. The first application shows how to derive the relation between a classical spin model on a graph and the corresponding model on the dual graph. The second application shows the implications of quantum mechanical symmetries existing in our description of classical systems by means of a quantum system.

Finally, we investigate the possibility to use simulations of the quantum system on a classical computer in order to obtain the statistics of  quantum measurement results. This investigation yields some insight into the complexity of the computation of the partition function and correlation functions of the classical system. We give a sufficient criterion for the structure of the interaction graph of the classical model, such that the computation of the partition function and correlation functions scale polynomially with system size.

\subsection{Duality relations for planar graphs}

We review Ref.~\cite{VDB06}. From graph theory it is known that  for any planar graph $G$ we can construct its dual graph $D$. In this section we want to demonstrate that the partition function $Z_{G}$ of a classical spin model defined on the graph $G$ and the partition function $Z_{D}$ of the model derived on the corresponding dual graph $D$ have a simple and meaningful relation.

To show this, we note that any orientation $\sigma$ of a graph $G^{\sigma}$ induces an orientation of its dual graph $D$~\cite{GR01}, which we also denote by $\sigma$ (we refer to~\cite{GR01}, page 168 for details). Moreover the incidence matrices $B\left(D^{\sigma}\right)$ and $B\left(G^{\sigma}\right)$ corresponding to the two graphs have the property $B\left(G^{\sigma}\right)B\left(D^{\sigma}\right)^T=0$ and the spaces generated by the rows of these matrices are each others duals $C_{G}\left(q\right)^{\perp}=C_{D}\left(q\right)$ . Hence, the stabilizer of $\left|\psi_{D}\right\rangle $ can be written as
\[
\mathcal{S}_{\left|\psi_{D}\right\rangle }
=\left\{ X\left(v\right)Z\left(u\right)|v\in C_{D}
\left(q\right),u\in C_{G}\left(q\right)\right\} .
\]
The quantum Fourier transform,
\[
F:=\frac{1}{\sqrt{q}}\sum_{j,k=0}^{q-1}
e^{\frac{2 \pi i k j}{q}}\ket{j}\! \bra{k},
\]
has the property to map $X$ and $Z$ to each other under conjugation: $FXF^{\dagger}=Z$ and $FZF^{\dagger}=X$, and can accordingly be used to map $\mathcal{S}_{\left|\psi_{D}\right\rangle }$ to  $\mathcal{S}_{\left|\psi_{G}\right\rangle }$, one-to-one, since  $F^{\otimes N}X\left(v\right)Z\left(u\right) \left(F^{\otimes N}\right)^{\dagger}=Z\left(v\right)X\left(u\right).$ Considering the identity
\[
\rho_{\mathcal{S}}=\frac{1}{q^{N}}\sum_{g\in\mathcal{S}}g
\]
for the density matrix $\rho_{\mathcal{S}}$ of a stabilizer state that is stabilized by the $q^{N}$ operators in $\mathcal{S}$, we infer that
\[
\left|\psi_{D}\right\rangle =F^{\otimes N}\left|\psi_{G}\right\rangle.
\]
The corresponding partition function $Z_{G}$ can thus be rewritten
as
\[
\bra{\psi_{G}} \left(\bigotimes_{e\in E} \ket{\alpha_{e}} \right) = \bra{\psi_{D}} \left(\bigotimes_{e\in E} \ket{\alpha'_{e}}\right),
\]
where $\ket{\alpha'_{e}} = F^{\dagger}\ket{\alpha_{e}}$. This transformation carries over to the energy terms in the Hamiltonian function of the model on the dual graph, where we find $Z_{G}\left(q,\sigma, \left\{ h_{e}\right\} \right)=Z_{D}\left(q,\sigma,\left\{ h_{e}'\right\} \right)$ with new energy terms $h_{e}'$, which are derived from the old ones by
\[
e^{-\beta h_{e}'\left(j\right)}:=\frac{1}{\sqrt{q}}\sum_{k=0}^{q-1}
e^{-\frac{2\pi i k j}{q}}e^{-\beta h_{e}\left(k\right)}
\]
for every $j=0,...,q-1$.

We now want to examine the relation of the the Potts model on a graph $G$ without external field and its corresponding model on the dual graph $D$. The Potts model, characterized by the Hamiltonian function
\[
H\left(\left\{ s_{i}\right\} \right)
=-\sum_{e=\bigl\langle i,j\bigr\rangle}J_{e}\delta_{ij},
\]
is encoded in two quantum states, $\left|\psi_{G}\right\rangle $ and $\bigotimes_{e\in E}\left|\alpha_{e}\right\rangle $ with
\[
\left|\alpha_{e}\right\rangle 
=\left|\alpha\right\rangle _{\mbox{\tiny{Potts}}}
=e^{\beta J_{e}}\left|0\right\rangle +\sum_{j=1}^{q-1}\left|j\right\rangle.
\]
The application of $F^{\dagger}$ on $\left|\alpha_{e}\right\rangle $ yields
\[
q^{1/2}e^{-\beta h'_{e}\left(j\right)}=\begin{cases}
e^{\beta J_{e}}+q-1 & \mbox{if }j=0\\
e^{\beta J_{e}}-1 & \mbox{if }j=1,...,q-1.\end{cases}
\]
Since the energies are again the same for all $j=1,...,q-1$, we have another Potts model (on the dual graph $D$) whose interaction strength $J'_{e}$ fulfills the relation
\[
e^{\beta J'_{e}}:=\frac{e^{\beta J_{e}}+q-1}{e^{\beta J_{e}}-1}.
\]
Equivalently, we write $\left(e^{\beta J'_{e}}-1\right) \left(e^{\beta J_{e}}-1\right)=q,$ and hence recover the well known high-low temperature duality relation for the Potts model partition function~\cite{Wu84}.

\subsection{Local symmetries}

See Ref.~\cite{VDB06}. Local symmetries of stabilizer states can be used to show that several different models of classical spin systems actually have the same partition functions. More precisely, any local unitary $U=\bigotimes_{e}U_{e}$ operator with eigenstate $\ket{\psi_{G}}$ \begin{equation} U\ket{\psi_{G}}=\lambda\ket{\psi_{G}}\label{eq:local-unitary-eigenstate}\end{equation} generates a model with the same interaction pattern but modified interaction strengths. Using Eq.~\eqref{eq:local-unitary-eigenstate} we obtain the symmetry relation
\[
(\bigotimes_{e\in E} \bra{\alpha_{e}}) \ket{\psi_{G}}=(\bigotimes_{e\in E}
\bra{\tilde{\alpha}_{e}})\ket{\psi_{G}}
\]
where
\[
\bigotimes_{e\in E} \ket{\tilde{\alpha}_{e}}
=\lambda^{*}U\bigotimes_{e\in E}\ket{\alpha_{e}}.\]
The mapping
\[
\sum_{j=0}^{q-1}e^{-\beta h_{e}\left(j\right)}\left|j\right\rangle 
=\left|\alpha_{e}\right\rangle \mapsto U_{e}\ket{\alpha_{e}}
=\sum_{j=0}^{q-1}e^{-\beta h_{e}\left(j\right)}U_{e}\left|j\right\rangle
\]
implies another mapping of the energies defining the prefactors of the basis states $\ket{j}$. This can lead to unphysical interaction strengths, e.g., imaginary ones.

Similarly, a relation for the states $\ket{\varphi_{G}}$ can be found, where the local symmetry is now corresponding to a change of interaction strengths and local field strengths
\[
(\bigotimes_{v\in V} \bra{\alpha'_{v}} \bigotimes_{e\in E} \bra{\alpha_{e}}) \ket{\varphi_{G}} = (\bigotimes_{v\in V} \bra{\tilde{\alpha}'_{v}} \bigotimes_{e\in E} \bra{\tilde{\alpha}_{e}}) \ket{\varphi_{G}},
\]
where
\[
\bigotimes_{v\in V} \ket{\tilde{\alpha}'_{v}} \bigotimes_{e\in E}
\ket{\tilde{\alpha}_{e}} = \lambda^{*}U \bigotimes_{v\in V}
\ket{\alpha'_{v}} \bigotimes_{e\in E} \ket{\alpha_{e}}.
\]

The effect on the correlation function is again similar, but generically different correlation functions will, by the same symmetry transformation, be mapped to the corresponding correlation functions of different models. By definition, the state $\ket{\alpha}$ enabling us to read out the value $\left\langle s_{i_{1}},s_{i_{2}},...,s_{i_{n}} \right\rangle _{\beta}$ is $\bigotimes_{v\in V}\ket{\alpha'_{v}} \bigotimes_{e\in E}\ket{\alpha_{e}}$ with
\begin{eqnarray*}
\ket{\alpha_{e}} & = & 
\sum_{j=0}^{q-1}e^{-\beta h_{e}\left(j\right)}\ket{j}\\
\ket{\alpha'_{v}\left(i_{1},...,i_{n}\right)}& = & 
\sum_{j=0}^{q-1}\cos\left(2\pi j/q\right)^{m_{\nu}}
e^{-\beta b_{v}\left(j\right)}\ket{j},
\end{eqnarray*}
where $m_{\nu}$ is the number of occurrences of $\nu$ in the n-tuple $\left(i_{1},...,i_{n}\right)$. Now
\[
U_{\nu}\ket{\alpha'_{v}\left(i_{1},...,i_{n}\right)}
=U_{\nu}\sum_{j=0}^{q-1}
\cos\left(2\pi j/q\right)^{m_{\nu}}e^{-\beta b_{v}\left(j\right)}\ket{j},
\]
so in general not only $h_{e}\left(j\right)$ and $b_{v}\left(j\right)$ will be altered, but the prefactors $\cos\left(2\pi j/q\right)^{m_{\nu}}$ play the role of weights. These are specific for the correlation function in question and enter the calculation of the energy terms $b_{\nu}\left(j\right)$ belonging to the symmetry.

The fact that the states $\ket{\psi_{G}}$ and $\ket{\varphi_{G}}$ are stabilizer states is advantageous, because all elements from the stabilizer define such a symmetry operation already, which we will use in the following examples.

\begin{example}
We consider now the change of a classical model with $q=2$ encoded into a state $\ket{\psi_{G}}$, caused by a symmetry operation. Let the classical graph have a vertex $a$ with a set of edges $E_{a}$ connecting to it. One column $c_{a}$ of the incidence matrix corresponds to the vertex $a$. The stabilizer element $X\left(c_{a}\right)Z\left(0\right)$ applied to the state $\ket{\alpha}$ (encoding the interaction strengths) maps all interactions strengths $J_{e},e\in E_{a}$ to $-J_{e}$ and does not touch the other ones. We hence obtain the result that
\[
Z\left(\left\{ J_{e}\right\} \right)=Z\left(\left\{ \tilde{J}_{e}\right\} \right),
\]
where
\[
\tilde{J}_{e}=
\begin{cases}
-J_{e} & e\in E_{a}\\
J_{e} & \mbox{otherwise}
\end{cases}.
\]

Next, we consider the change of a classical model with $q=2$ encoded into a state $\ket{\varphi_{G}}$, which is caused by a symmetry operation. The matrix generating the cutspace is now $C=\left(\one|B\right)^T$. The construction of the local unitary symmetry operation using one column of $C$, like in the example above, yields now
\[
Z\left(\left\{ b_{\nu},J_{e}\right\} \right)
=Z\left(\left\{ \tilde{b}_{\nu},\tilde{J}_{e}\right\} \right),
\]
where
\[
\tilde{J}_{e}=
\begin{cases}
-J_{e} & e\in E_{a}\\
J_{e} & \mbox{otherwise}
\end{cases}
\]
and $\tilde{b}_{a}=-b_{a}$ and $\tilde{b}_{\nu}=b_{\nu}$ otherwise.
\end{example}

\subsection{\label{sub:Simulations-on-classical}Simulations on classical computers}

An interesting aspect of the proposed mapping from classical to quantum systems is the established link between two different mathematical formalisms. As shown, algorithms for the computation of overlaps of stabilizer states with product states can be used to compute partition sums and correlation functions of classical spin systems -- and vice versa. In both cases, hard and computationally feasible instances of these calculations are known, and we can now extend efficient algorithms from one domain to the other. This connection allows us to prove the following

\begin{thm}
\label{thm:simul_comprehensible}There exists an algorithm that allows one to compute the partition function and the correlation functions of classical spin models defined on graphs exactly and with an effort that scales polynomially in the number of spin sites, provided that the tree-width of the graph used to define the classical model scales logarithmically in the number of spin sites.
\end{thm}
The proof is rather technical and is given in  appendix~\ref{sec:Proof-of-Simultability}. Thus, one finds that partition  functions on graphs which are sufficiently  similar to a tree graph (a property made precise by the notion of tree-width)  can be efficiently evaluated. Similar results have been obtained in, e.g.,  Refs.~\cite{MS06}.

\subsection{Relations to measurement based quantum computation}

In this section we discuss how the mappings between classical spin systems and the quantum stabilizer formalism presented in this work, may provide insights in the study of measurement-based (or ``one-way'') quantum computation (MQC).

The one-way quantum computer is a model of quantum computation introduced in Ref.~\cite{RB01}. In contrast to the quantum circuit model, where quantum computations proceed by unitary evolutions, in MQC any computation is realized via single-qubit measurements only. More precisely, a one-way quantum computation essentially consists of two main steps: first, a system of many qubits is prepared in a highly entangled state, the ``2D cluster state''~\cite{BR01}, which is an instance of a stabilizer state. Second, part (possibly all) of the qubits in the system are measured individually. The qubits are measured one after the other in a specific order, and each qubit is measured in a certain basis which may (and typically does) depend on the outcomes of previous measurements. It is this ``measurement pattern'' which determines the quantum algorithm which is implemented.

It was shown in Refs.~\cite{RB01, RBB03} that the one-way quantum computer is a universal model for quantum computer, i.e., it is capable of (efficiently) simulating every quantum computation performed within the quantum circuit model. We refer to Ref.~\cite{RBB03} for more details about MQC.

Note that the model of MQC exhibits a remarkable feature, namely that the entire resource of a quantum computation is carried by the entangled cluster state in which the system is initially prepared. Indeed, as local measurements can only destroy entanglement, all the entanglement present within a one-way quantum computation must be provided by the initial resource state. Therefore, in order to understand the computational power of quantum computers, a study of the properties of 2D cluster states, and other resource states, is called for.

Even though it is by now well-established that the 2D cluster states are universal resource states for MQC (and several other states have also been found to be universal~\cite{Va06a,GE07}, it is not yet fully understood which properties of these states are responsible for their universality. This issue has been the topic of recent investigations~\cite{Va06a,Brprep} (see also \cite{BR06,NGESP06,VDB07,Moprep}), where it was studied under which conditions a given quantum state may be a universal resource for MQC, and under which conditions it does \emph{not} provide any computational speed-up with respect to classical computation. While significant progress has been made in these works, this important problem is far from being fully understood.

What can the present connections between classical spin systems and quantum stabilizer states teach us about MQC? To this end, consider a one-way computation having one of the stabilizer states $|\varphi_{G}\rangle$ or $|\psi_{G}\rangle$ as a resource, where $G$ is some graph. One may then ask which computational power can such resource states provide for MQC -- i.e., which states among the $|\varphi_{G}\rangle$ and $|\psi_{G}\rangle$ are universal resource states, and which states are fully simulatable classically. Next we will see how the relation between these quantum states and the associated classical spin systems, as established in this paper, provides insights in this issue.

To do so, consider Eq.~\eqref{eq:Zeqoverlap}, which identifies overlaps between a resource state $|\eta_G \rangle$ ($\equiv|\psi_{G}\rangle$ or $|\varphi_{G}\rangle$) and a product state $|\alpha\rangle$, as the partition function $Z_{G}$ of the associated classical spin model on the graph $G$. Now note that such overlaps (to be precise, their squared modulus) equal the probabilities of outcomes of local measurements performed on the resource state $|\eta_G\rangle$. Therefore, if it is possible to compute such overlaps (and thus the corresponding measurement probabilities) efficiently, it becomes possible to simulate local measurement processes on such a resource, on a classical computer. Resources for which such efficient classical simulation is possible, by definition cannot offer any computational speed-up as compared to classical computation. Using Eq.~\eqref{eq:Zeqoverlap}, we now see that the problem of computing measurement probabilities of local measurements boils down to the evaluation of the partition function of the associated classical model. In particular, we find that \emph{classical models which are ``solvable''---i.e., their partition function can be efficiently evaluated---give rise to resource states for which the associated probabilities of local measurements can be computed efficiently}. Therefore, the present mappings establish a relation between the solvability of a classical spin systems and the computational power of the associated resource state.

Let us illustrate these relations with some examples for Ising models on different lattice types, with or without magnetic fields (see also Figs.~\ref{fig:nondecorated} and \ref{fig:decorated}). Consider e.g., the simple case of a 1D Ising model with periodic boundary conditions, without external field. This model is known to be solvable: its partition function can be evaluated in a time which scales polynomially with the number of spins. Using our correspondence, the associated quantum state $|\psi_{G}\rangle$ is a GHZ state (see example \ref{exa:StabAndProdPsiGee}). This state is known to be an efficiently classically simulatable resource state for MQC. A similar conclusion can be drawn for the 1D Ising model in the presence of an external field, which is solvable as well. Using our mappings, the associated quantum state $|\varphi_{G}\rangle$ is a 1D cluster state, which is indeed also known to be simulatable (see, e.g.,~\cite{VDB07}). Finally, also the 2D Ising model without field is known to be solvable -- this is Onsager's famous result. The corresponding stabilizer state $|\psi_{G}\rangle$ is the toric code state. And indeed, this state is a simulatable resource -- in fact, the latter property has been shown in Ref.~\cite{BR06} by \emph{using} the relation between this state and the solvable 2D Ising model.

An Ising model which is not solvable is the 2D Ising model in the presence of an external field. In fact, the evaluation of its partition function is an NP-hard problem. The corresponding stabilizer state is the 2D (decorated) cluster state. Interestingly, this state is a \emph{universal} resource for MQC. Therefore, we find that also in this case the computational difficulty of a classical model is reflected in the quantum computational power of the associated quantum state.

\section{\label{sec:Summary-and-Conclusion}Summary and Conclusion}

In this work, we have displayed several mappings from Hamiltonian functions of classical spin systems to states of quantum spin systems. We map the interaction pattern given by the Hamiltonian function of the classical system to quantum stabilizer states and the interaction strengths as well as local field strengths to quantum product states. The overlap of these states yields the macroscopic quantities of the thermal states of the classical spin system: the partition function and correlation functions at freely selectable temperatures (which are also encoded into the product states).

The described mappings circumfere different classes of admissible Hamiltonian functions. From the original and exemplary approach~\cite{VDB06} suited for two-body interactions without local fields, we derive a more generalized mapping capable to yield correlation functions as well as to include local fields. Finally, we introduce a version capable to treat arbitrary Hamiltonian functions with $n$-body terms. Each of these mappings is interesting in its own right and offers an individual viewpoint and individual aspects in the formal approach. The relations between the different mappings were investigated.

We moreover gave several applications of the proposed mappings, namely: a simple derivation of the duality relation of a graph and its dual; a simple derivation of the impact of local symmetries of the stabilizer state on the classical model described by it; a constructive proof of a sufficient criterion for the possibility to efficiently evaluate of the thermal quantities of a classical spin system on a classical computer; and we discussed the relation of the computational accessibility of a classical spin system with the power of a quantum computer.

\acknowledgments{This work was supported by the FWF and the European Union (QICS, OLAQUI, SCALA). MVDN acknowledges support by the  excellence cluster MAP. The authors thank G.~Ortiz, M.A.~Martín-Delgado  and G.~De las Cuevas for discussions.}

\appendix

\section{A proof for the given number of stabilizer elements\label{sec:NumElProof}}

In section \ref{sub:The-basic-principle} we constructed the set of operators $X\left(v\right)Z\left(u\right)$ [see Eq.~\eqref{eq:XvZu}] with $v\in C_{G}\left(q\right)$ and $u\in C_{G}\left(q\right)^{\perp}$, where by construction $C_{G}\left(q\right)$ is the $\mathbb{Z}_{q}$-sub-module of $\mathbb{Z}_{q}^{N}$ that is generated by the rows of the incidence matrix $B^{\sigma}$. For this set to be a stabilizer of the single state $\left|\psi_{G}\right\rangle $ it is necessary that $\left|\psi_{G}\right\rangle $ is a fixed point of these operators (as already shown in the indicated section) and that is has cardinality $q^{N}$. The latter point we show now.

\begin{lem}
The number of independent operators generated by $X\left(v\right)Z\left(u\right)$ [see Eq.~\eqref{eq:XvZu}] with $v\in C_{G}\left(q\right)$ and  $u\in C_{G}\left(q\right)^{\perp}$ is $q^{N}$. \end{lem}
\emph{Proof.} We note that the module $\mathbb{Z}_{q}^{N}$ and hence also all its sub-modules are free modules. Accordingly we can chose a basis, from which the modules or sub-modules are generated respectively. With the scalar product $\left\langle \cdot |\cdot \cdot \right\rangle $ we construct an orthonormal basis $\left\{ c_{i}\right\} $ and with it the following mapping
\[
\varphi:\mathbb{Z}_{q}^{N}\rightarrow\mathbb{Z}_{q}^{N},w
\mapsto\sum_{c_{i}\in C_{G}\left(q\right)}c_{i}\braket{c_{i}}{w}.
\]
This is a module-homomorphism, since for $\lambda,\mu\in\mathbb{Z}_{q}$ and $a,b\in\mathbb{Z}_{q}^{N}$
\begin{multline*}
\varphi\left(\lambda a+\mu b\right)
=\sum_{c_{i}\in C_{G}\left(q\right)} c_{i}\braket{c_{i}}{\lambda a+\mu b}=\\
\lambda \sum_{c_{i}\in C_{G}\left(q\right)} c_{i} \braket{c_{i}}{a} +\mu
\sum_{c_{i}\in C_{G}\left(q\right)} c_{i} \braket{c_{i}}{b} =\lambda\varphi\left(a\right)+\mu\varphi\left(b\right),
\end{multline*}
by the linearity of the scalar product. The kernel of $\varphi$, $\mathrm{ker}\left(\varphi\right)$, is the set $C_{G}\left(q\right)^{\perp}$ because being a (orthonormal) basis $\left\{ c_{i}\right\} $ is independent. The range of $\varphi$, $\mathrm{ran}\left(\varphi\right)$, is the set $C_{G}\left(q\right)$, because for every $w\in C_{G}\left(q\right)$ we have $w=\sum_{i}\lambda_{i}c_{i}$ and $\varphi\left(w\right)=\sum_{c_{i}\in  C_{G}\left(q\right)}\sum_{j}\lambda_{j}c_{i}\braket{c_{i}}{c_{j}} =w$. The homomorphism $\varphi$, as any module-homomorphism induces an isomorphism
\[
\mathbb{Z}_{q}^{N}/\mathrm{ker}\left(\varphi\right)
\tilde{\longrightarrow}\mathrm{ran}\left(\varphi\right),
\]
which provides us with the formula
\[
\frac{\bigr|\mathbb{Z}_{q}^{N}\bigl|}{\bigl|C_{G}\left(q\right)^{\perp}\bigr|}
=\frac{\bigr|\mathbb{Z}_{q}^{N}\bigl|}{\bigl|\mathrm{ker}\left(\varphi\right)\bigr|}
=\bigl|\mathrm{ran}\left(\varphi\right)\bigr|=\bigl|C_{G}\left(q\right)\bigr|
\]
relating the number of elements in these sets. This implies
\begin{equation}
q^{N}=\bigl|C_{G}\left(q\right)^{\perp}\bigr|\bigl|C_{G}\left(q\right)\bigr|.
\label{eq:qnccperp}
\end{equation}
The number on the r.h.s.\ equals the number of the constructed operators $X\left(v\right)Z\left(u\right)$, which are, as a set, isomorphic to
\[
\left\{ \left(c,s\right)|c\in C_{G}\left(q\right),s\in 
C_{G}\left(q\right)^{\perp}\right\}.
\]
This concludes the proof. $\hfill\square$

\section{\label{sec:TTNtheory}Tensor tree networks and tensor tree states}

We follow an approach of Shi, Duan and Vidal and consider the description of states in terms of a tensor network with tree structure~\cite{SDV06,MS06}. We now want to give a short overview of fundamental definitions and theorems concerning these tensor tree states (TTS).

\subsection{Basic definitions}

The building block of a tensor network are complex $d_{1}\times d_{2}\times...\times d_{n}$ tensors with elements $A_{i_{1}i_{2}...i_{n}}$. The number $n$ is called the \emph{rank} of the tensor $A$ and the number $d_{k}$ is called the \emph{rank of the index} $i_{k}$. The maximal number $d$ that the indices can assume, $d=\max_{k}d_{k}$, is called the \emph{dimension} of the tensor. A summation over two indices $i_{l}$ and $j_{l'}$ of common rank of two tensors $A^{[r]}$ and $A^{[s]}$,
\begin{multline*}
A_{i_{1}i_{2}...\hat{i_{l}}...i_{n}j_{1}j_{2}...\hat{j_{l'}}...j_{n'}}^{[r,s]}\\
=\sum_{k}A_{i_{1}i_{2}...(i_{l}=k)...i_{n}}^{[r]}A_{j_{1}j_{2}...(j_{l'}=k)...j_{n'}}^{[s]},
\end{multline*}
is called a contraction of the indices $i_{l}$ and $j_{l'}$. A set of tensors together with pairs of indices that are to be contracted is called a \emph{tensor network}. The maximal dimension $D$ of all tensors, $D=\max_{A}d[A]$, is called the \emph{dimension of the network}. Tensor networks can be represented by graphs: the each vertex of the graph corresponding to one tensor of the network and each edge corresponding to one pair of contracted indices. The indices to be contracted are referred to as \emph{internal} indices and the other ones as \emph{open}. The notation of graph theory carry over to the tensor networks, e.g., we talk about ``subcubic'' tensor trees. A tree graph (network) is called \emph{subcubic} if each vertex (tensor) has degree (rank) 1 or 3. The vertices with rank 1 are called \emph{leaves}.

It is possible to write the coefficients $A_{\mathbf{s}}$ of a generic pure $N$-qudit state $\left|\varphi\right\rangle = \sum_{\mathbf{s}}A_{\mathbf{s}}\left|\mathbf{s}\right\rangle $, where $\left\{ \left|\mathbf{s}\right\rangle \right\} $ is a product basis, as a contraction of a fixed set of tensors. Trivially, one tensor of rank $N$ and a dimension equal to the number of states of the qudits is sufficient. In fact, representations for any graph-structure can be found, provided the rank of the internal indices being sufficiently large. Depending on the internal structure of the state to be represented, even representations with internal indices of comparatively small rank might be found, hence reducing the number of complex parameters representing the network. This displays the principle that the more structure there is in the state, the less information is (potentially) needed to settle the remaining degrees of freedom. Conversely, any tensor network with $N$ open indices can be used to define a pure $N$-qubit state.

\begin{figure*}
\includegraphics[width=0.8\columnwidth]{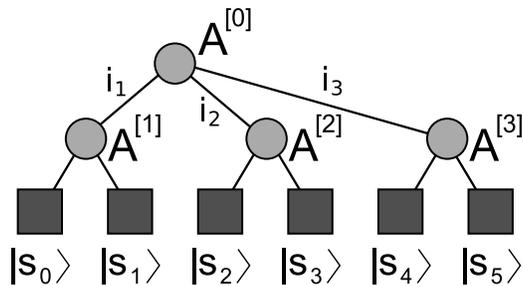}

\caption{\emph{A subcubic tensor network.} The vertices correspond to  tensors (circles) or physical sites (squares, ``leaves'') respectively. Edges indicate contractions over common indices. The bipartition of a state corresponding to the cut of a single vertex cannot have any Schmidt-rank higher than the rank of the connecting indices. The class of states generated by tensor networks covers all possible pure states provided that the dimension of the network is sufficiently large.} 

\label{fig:trees}
\end{figure*}

As an illustrative example, the tree depicted in Fig.~\ref{fig:trees}~a) corresponds to the state
\[
\ket{\tau}=\sum_{\mathbf{s}}\sum_{i_{j}}A_{i_{1}i_{2}i_{3}}^{0}A_{i_{1}s_{0}s_{1}}^{1}
A_{i_{2}s_{2}s_{3}}^{2}A_{i_{3}s_{4}s_{5}}^{3}\ket{\mathbf{s}}.
\]
Another well known example of subcubic tensor tree states are the matrix product states (MPS) with open boundary conditions. They have the simple form
\begin{multline*}
\ket{\mbox{MPS}}=\sum_{\mathbf{s}}\sum_{i_{j}}A_{i_{0}s_{0}s_{1}}^{0}
\ket{s_{0}s_{1}}A_{i_{0}i_{1}s_{2}}^{1}\ket{s_{2}}A_{i_{1}i_{2}s_{3}}^{2}\ket{s_{3}}\\
\times...A_{i_{N}s_{N-1}s_{N}}^{N-1}\ket{s_{N-1}s_{N}}.
\end{multline*}

\subsection{Efficient scaling}

A subcubic TTN with $N$ open indices (representing a subcubic TTS of $N$ qudits) and dimension $D$ depends on at most $O\left(ND^{3}\right)$ complex parameters. Thus, a family of states over $N$ qudits whose TTN-description has a dimension scaling polynomially in $N$ allows for a description with a number of parameters scaling polynomially in $N$. Concerning the contractions of TTS with product states we obtain the following result, (see also Refs.~\cite{SDV06,MS06})

\begin{lem}
\label{lem:overlap_complexity}Calculating the overlap of a complete product state of $N$ qudits with a subcubic TTS of dimension $D$ over $N$ qudits has a complexity of at most $O\left(ND^{3}\right)$.
\end{lem}
\emph{Proof.} Let the product state be 
\[
\ket{\alpha}=\bigotimes_{l\in\mbox{\tiny{leaves}}}\ket{\alpha_{l}}
\]
and the TTS be $\ket{\tau}$. The calculation of $\braket{\alpha}{\tau}$ is a contraction of a subcubic tensor network where the leaves are tensors with values $\braket{\alpha_{l}}{l}$. In a subcubic tree, there is at least one tensor with at least two leaves attached. A contraction of this tensor with its attached leaves requires an effort of order $D^{3}$. If this tensor has three leaves attached we are done. If not, this tensor will now be a leaf tensor attached to one other tensor and the tree will still be subcubic. As before, the tree will now have at least one tensor which has at least two leaves attached. We continue this procedure and because there are $N-2$ tensors in the tree, we end up with an effort of the order $ND^{3}$.$\hfill\square$

Hence the contraction of a family of states over $N$ qudits, whose TTN-description has a dimension scaling polynomially in $N$, with product states of the appropriate Hilbert spaces scales polynomially in $N$.

\subsection{Entanglement in TTS}

We want to state one more important result concerning the entanglement content of TTS with dimension $D$. Since the rank of the index corresponding to an edge connecting two tensors is limited by this number $D$, only $D$ linearly independent combinations of states corresponding to the sub-trees attached to this edge are possible. Hence we have

\begin{lem}
\label{lem:limited_schmidt_number}The number of Schmidt coefficients of a TTS with dimension $D$ in a bipartition of the qudits that corresponds to cutting exactly one edge in the graph of the corresponding TTN is limited by $D$. This Schmidt number can be reached. $\hfill\square$
\end{lem}

\section{\label{sec:Proof-of-Simultability}A proof of theorem 
\ref{thm:simul_comprehensible}}

We will prove theorem \ref{thm:simul_comprehensible} using the machinery developed in the sections \ref{sec:Encoding-classical-spin}  and \ref{sec:Extending-the-formalism}. The underlying idea of the proof is to map the classical spin problem (of finding the partition function) to the corresponding quantum problem (of finding an overlap), which is then solved by a simulation on a classical computer. To treat the simulation aspect, we need some results from the theory of tree tensor networks~\cite{SDV06,MS06}. We will use the tree tensor networks to encode the stabilizer states which are the images of the interaction patterns of the classical spin systems. The necessary notation and theorems have been summarized in the appendix \ref{sec:TTNtheory}. With the language developed there, we reformulate theorem \ref{thm:simul_comprehensible}.

\subsection{Simulation complexity for the states $\left|\varphi_{G}\right\rangle $ and $\left|\psi_{G}\right\rangle $}

\begin{thm}
\label{thm:scaling}For the states $\ket{\psi_{G}}$ and $\ket{\varphi_{G}}$ (as defined above) a tree tensor network description can be computed with an effort growing polynomially in the number of classical spin sites $N$, provided that the tree-width of the graph $G$ grows logarithmically in $N$. This tree tensor network description allows to compute the overlaps $\braket{\alpha}{\psi_{G}}$ and $\braket{\alpha'}{\varphi_{G}}$ of these states with product states with an effort that grows polynomially in $N$.
\end{thm}
\emph{Proof.} The proof consists of three parts. i) In a preparatory step, we will summarize the ties between the tree-width of $G$ and the branch-width of its cycle matroid. ii) We will then use this result to derive a bound for the Schmidt-rank of a TTN-description of the states $\ket{\psi_{G}}$ and $\ket{\varphi_{G}}$ and hence derive an upper bound for the computational effort to compute the overlaps $\braket{\alpha}{\psi_{G}}$ and $\braket{\alpha'}{\varphi_{G}}$. iii) Finally we give an algorithm to find the (tensor-)coefficients in the TTN-description. Parts ii) and iii) have been given in a similar form for $q=2$ already in Ref.~\cite{VDB07}.

i) Let us first fix some notation, which can be found in more detail, together with missing proofs, for example in the references~\cite{GR01,HW06,RS84}. 

\begin{defn}
(matroid) A \emph{matroid} is a set $\Omega$ together with a rank function $rk$ on its subsets. A \emph{rank function} fulfills the following properties
\end{defn}
\begin{itemize}
\item If $A$ and $B$ are subsets of $\Omega$ and $A\subset B$, then $rk\left(A\right)\leq rk\left(B\right)$. \item For all subsets $A$ and $B$ of $\Omega$,
\[ 
rk\left(A\cap B\right)+rk\left(A\cup B\right)\leq rk\left(A\right)+rk\left(B\right).
\]

\item If $A\subset\Omega$, then $rk\left(A\right)\leq\left|A\right|.$
\end{itemize}

A special instance of a matroid is the \emph{set of columns} of the incidence matrix $B^{\sigma}$ of a graph $G=\left(V,E\right)$, called the \emph{cycle matroid} $M\left(G\right)$ of the graph $G$. 

A natural choice of a rank function on a cycle matroid is the dimension of the span of the column vectors. Following the ideas of lemma \ref{lem:prefactor} we deduce that with this choice of rank function and for a subset of column vectors $A$ we have the relation $rk\left(A\right)=\left|V\right|-c$, where $c$ is the number of connected components in the graph $G_{A}=\left(V,A\right)$.

\begin{defn}
(connectivity function) With the rank function $rk$ of the cycle matroid $M\left(G\right)$ of the graph $G=\left(V,E\right)$ we define the \emph{connectivity function $\lambda$ on a subset of edges $A\subset E$
by} 
\[
\lambda\left(A\right):=rk\left(A\right)+rk\left(E-A\right)-rk\left(E\right)+1.
\]

\end{defn}
It is a symmetric function with respect to $A\leftrightarrow E-A$. An important observation is that with $rk\left(A\right)=\left|V\right|-c(A)$ follows the equality
\begin{equation}
\lambda\left(A\right):=\left|V\right|+c\left(E\right)-c\left(A\right)
-c\left(E-A\right)+1,\label{eq:lambda-and-c}
\end{equation}
where $c\left(E\right),c\left(A\right)$ and $c\left(E-A\right)$ are the numbers of the connected components in the respective subsets of edges.

\begin{defn}
(branch decomposition) Let $T$ be a subcubic tree (see appendix \ref{sec:TTNtheory}
and Fig.~\ref{fig:trees}) with edges $E$. The deletion of an edge $e\in E$ of the tree corresponds to a bipartition of the set of leaves of the tree, because the deletion divides the tree into two (connected) components. The set of bipartitions of leaves induced by a tree is called a \emph{branch decomposition} of the leaves.
\end{defn}
In the following we will identify the edges of our (decorated) interaction graph $G$ describing the classical spin system with the leaves of a suitable tree. The set of possible trees then corresponds to a set of different branch decompositions. With the connectivity function we can define the branch width of a branch decomposition.

\begin{defn}
(branch-width) The \emph{branch-width $b_{T}\left(\lambda\right)$} associated with the branch decomposition induced by a tree \emph{$T$} with edges $E_{T}$ is the value
\[
b_{T}\left(\lambda\right):=\max_{e\in E_{T}}\lambda\left(A_{e}\right),
\]
where $A_e \subset E(T)$ is the subset of edges belonging to one of the  remaining contiguous sub-trees of tree $T$ obtained by deleting  edge $e \in E(T)$ from tree $T$ [note the  symmetry $\lambda(A_e)=\lambda(T-A_e)$.] The \emph{branch-width of the cycle matroid} of $G$ is defined as
\[
b\left(\lambda\right):=\min_{T}b_{T}\left(\lambda\right).
\]

With this notation, we formulate the following theorems, to be found, together with the proofs, in Ref.~\cite{HW06}.
\end{defn}
\begin{lem}
(Theorem 3.2 in Ref.~\cite{HW06}) Let $G$ be a graph with at least one edge, and let $M\left(G\right)$ be the cycle matroid of $G$. Then the tree-width of $G$ equals the tree-width of $M\left(G\right)$.$\hfill\square$
\end{lem}
and

\begin{lem}
(Theorem 4.2 in Ref.~\cite{HW06}) Let $M$ be a matroid of tree-width $t$ and branch-width $b$. Then
\[
b-1\leq t\leq\mbox{max}\left(2b-2,1\right).
\]
$\hfill\square$
\end{lem}
In particular this result tells us that the tree-width is an upper bound to the branch-width. The next theorem, to be found in Ref.~\cite{OPhD}, now states that we can algorithmically compute a subcubic tree that at least comes close to the optimal tree.

\begin{lem}
\label{lem:tree-search}(Theorem 2.12 in Ref.~\cite{OPhD}) For given $k$, there is an algorithm as follows. It takes as input a finite set $E_{G}$ with $\left|E_{G}\right|\geq2$ [and the connectivity function $\lambda$]. It either concludes that $b\left(\lambda\right)>k$ or outputs a tree with $b_{T}\left(\lambda\right)\leq3k+1$. Its running time is $O\left(\delta\left|E_{G}\right|^{6}\log\left|E_{G}\right|\right)$, where $\delta$ is the time to compute $\lambda$. $\hfill\square$
\end{lem}
Because efficient algorithms to compute the tree-width of a graph $G$ (and hence with the lemma above, upper bounds for the branch-width) are known, we can assume to be able to input a $k>b\left(\lambda\right)$. This way we always end up with a tree $T$ with at most $b_{T}\left(\lambda\right)=3k+1$.

At the end of this part we know, given an interaction graph $G$ with tree-width $t$, that we can efficiently compute a branch decomposition over the set of edges such that the branch-width associated with this decomposition is smaller than or equal to $3t$.

ii) We now want to establish a link between the $\chi-$width associated with bipartitions of the states $\ket{\psi_{G}}$ and $\ket{\varphi_{G}}$ and the branch-width of the a tree inducing these bipartitions. To fix some notation, we define the matrix
\begin{equation}
M:=
\begin{cases}
B^{\sigma} & \mbox{for }\left|\psi_{G}\right\rangle \\
\left(\one_{\left|V\right|}|B^{\sigma}\right) & \mbox{for }\left|\varphi_{G}\right\rangle.
\end{cases}
\label{eq:Mdefinition}
\end{equation}
We recognize that $M$ is used to define the stabilizer of the respective states, because 
\[
C_{G}=\left\{ M^T \mathbf{s},\mathbf{s} \in\mathbb{Z}_{q}^{\left|V\right|}\right\}.
\]

\begin{lem}
\label{lem:chieins} Let $P\cup Q=E$ with $P\cap Q=\textrm{\O}$ be a bipartition of the edges in the interaction graph $G=\left(V,E\right)$ describing the interaction pattern of the classical spin system. Let $\left|\eta_{G}\right\rangle $ denote the quantum state whose stabilizer is constructed via $M$ (e.g.\ $\left|\varphi_{G}\right\rangle $ or $\left|\psi_{G}\right\rangle $). The Schmidt rank $\chi$ of the bipartition $\left|\eta_{G}\right\rangle =\sum_{i=1}^{\chi}\lambda_{i}\left| \eta_{i}^{P}\right\rangle \left|\eta_{i}^{Q}\right\rangle $, where $\left|\eta_{i}^{P}\right\rangle $ and $\left|\eta_{i}^{Q}\right\rangle $ are quantum states of the qudits corresponding to the edges in  $P$ and $Q$ respectively, satisfies the equality \[ \chi=q^{\left|V\right|+c\left(E\right)-c\left(P\right)-c\left(Q\right)},\]  where $c\left(E\right),c\left(P\right),c\left(Q\right)$ are the number of connected components in the graphs $\left(V,E\right),\left(V,P\right)$ and $\left(V,Q\right)$ respectively.
\end{lem}
\textbf{Proof.} Corresponding to the bipartition $P\cup Q=E$ we have a bipartition of the columns of the matrix $M$. After performing some (unimportant) permutation of the columns, the matrix $M$ takes the form $M=\left(M_{P}|M_{Q}\right)$. Let $\mathfrak{c}$ denote the number of columns of $M$ and $\mathfrak{p}$ and $\mathfrak{q}$ the number of columns of $M_{P}$ and $M_{Q}$ respectively.

Now let $\mathcal{S}$ be the stabilizer of the state $\left|\eta_{G}\right\rangle $ and $\mathcal{S}_{P}\subset\mathcal{S}$ the subset of operators $g$ that act trivially on the qudits belonging to the labels in $Q$. We define  $\bar{\mathcal{S}}_{P}:=\left\{ \mbox{Tr}_{Q}\left[g\right],g\in\mathcal{S}_{P}\right\} $. From the  theory of stabilizer states it is known that $\ket{\eta_{G}} \! \bra{\eta_{G}} =q^{-\mathfrak{c}}\sum_{g\in\mathcal{S}}g$
and hence
\begin{multline*}
\rho_{P}=\\
\mbox{Tr}_{Q}\left[\ket{\eta_{G}}\!\bra{\eta_{G}} \right]=\mbox{Tr}_{Q}\left[q^{-\mathfrak{c}}\sum_{g\in\mathcal{S}}g\right]
=q^{\mathfrak{q}-\mathfrak{c}}\sum_{g\in\bar{\mathcal{S}}_{P}}g\\
=q^{-\mathfrak{p}}\sum_{g\in\bar{\mathcal{S}}_{P}}g.
\end{multline*}
The factor $q^{\mathfrak{q}}$ comes in because the trace over all operators but $\one$ in the Pauli group is zero and $\mbox{Tr}_{Q}[\one]=q^{\mathfrak{q}}$. Furthermore, the stabilizer is a group, so we have the identity
\begin{multline*}
\left(\rho_{P}\right)^{2}=\\
q^{-2\mathfrak{p}}\sum_{g\in\bar{\mathcal{S}}_{P}}g
\sum_{h\in\bar{\mathcal{S}}_{P}}h=q^{-2\mathfrak{p}}
\sum_{g\in\bar{\mathcal{S}}_{P}}\sum_{h\in\bar{\mathcal{S}}_{P}}h
=\frac{\left|\bar{\mathcal{S}}_{P}\right|}{q^{2\mathfrak{p}}}\sum_{h\in\bar{\mathcal{S}}_{P}}h\\
=\frac{\left|\bar{\mathcal{S}}_{P}\right|}{q^{\mathfrak{p}}}\rho_{P}.
\end{multline*}
We define $r:=q^{\mathfrak{p}}/\left|\bar{\mathcal{S}}_{P}\right|$ and obtain $\left(r\rho_{P}\right)^{2}=r\rho_{P}$. Hence $r\rho_{P}$ is a projector and has (after a possibly necessary change of basis and reordering of rows) the form $r\rho_{P}=\mbox{diag}\left(1,...,1,0,...,0\right)$, or equivalently, $\rho_{P}=\mbox{diag}\left(r^{-1},...,r^{-1},0,...,0\right)$. Since $\mbox{Tr}\left[\rho_{P}\right]=1$, we have $r^{-1}\mbox{rank}\left(\rho_{P}\right)=1$ and hence $r$ equals the number of Schmidt coefficients in the bipartition of the state $\left|\eta_{G}\right\rangle $ according to the sets of edges $P$ and $Q$. Thus $\chi=r=q^{\mathfrak{p}}/\left|\bar{\mathcal{S}}_{P}\right|$.

To obtain the number $\left|\bar{\mathcal{S}}_{P}\right|$, we have now a look at the matrix $M=\left(M_{P}|M_{Q}\right)$, which we will from now on interpret as a linear mapping  $M^T:\mathbb{Z}_{q}^{\left|V\right|}\rightarrow\mathbb{Z}_{q}^{\mathfrak{c}}$. Here, $M_{P}$ is a $\left|V\right|\times\mathfrak{p}$-matrix belonging to the columns in $P$ and $M_{Q}$ is a $\left|V\right|\times\mathfrak{q}$-matrix belonging to the columns in $Q$. Recall that the stabilizer is
isomorphic to the set of operators
\[
X\left(v\right)Z\left(u\right)
:=\bigotimes_{c\in\mbox{\tiny{columns of }}M}X^{v_{c}}Z^{u_{c}},
\]
where $v\in C_{G}\left(q\right)$ and $u\in C_{G}\left(q\right)^{\perp}$. Hence $\left|\bar{\mathcal{S}}_{P}\right|$ is determined by the number of vectors $v'\subset C_{G}\left(q\right)$ and $u'\in C_{G}^{\perp}\left(q\right)$ whose elements are $0$ in the last $\mathfrak{q}$ places (e.g.\ $v'=\left(v'_{1},..., v'_{\mathfrak{p}},0,...,0\right)$). Let this  number for the set $C_{G}\left(q\right)$ be $z_{C}=\left|C_{P}\left(q\right)\right|$, where $\left(v'_{1},...,v'_{\mathfrak{p}}\right)\in C_{P}\left(q\right)$, and the corresponding number for the set $C_{G}\left(q\right)^{\perp}$ be $z_{C^{\perp}}=\left|C_{P}^{\perp}\left(q\right)\right|$, where $\left(u'_{1},...,u'_{\mathfrak{p}}\right)\in C_{P}^{\perp}\left(q\right)$. Then $\left|\bar{\mathcal{S}}_{P}\right|=z_{C}z_{C^{\perp}}$.

Let us now calculate $z_{C}$. The elements of $C_{G}$ are the image vectors of $M^T$. Furthermore, if $s\in\mbox{ker}\left(M_{Q}^T\right)$, then the image of $s$ has the desired form  $M^Ts=v'=\left(v'_{1},...,v'_{\mathfrak{p}},0,...,0\right)$. Considering that we can add any vector from the kernel of $M^T$ to $s$ without changing the image $v$, it is  $z_{C}=\left|\mbox{ker}\left(M_{Q}^T\right)\right| /\left|\mbox{ker}\left(M^T\right)\right|$.

Similarly, $z_{C^{\perp}}$ equals the number of elements in the set $C_{P}\left(q\right)^{\perp}$ where $C_{P}\left(q\right)=\mbox{ran}\left(M^T_{P}\right)$. As shown as part of Appendix A, this number is equal to  $z_{C^{\perp}}=q^{\mathfrak{p}}/\left|\mbox{ran}\left(M_{P}^T\right)\right|$ as the target space of the mapping $M_{P}^T$ is $\mathbb{Z}_{q}^{\mathfrak{p}}$.

Another basic consideration about the linear mapping  $M_{P}^T:\mathbb{Z}_{q}^{\left|V\right|}\rightarrow\mbox{ran}\left(M_{P}^T\right)$ (note: a mapping between finite spaces) tells us that
\[
q^{\left|V\right|}=\left|\mbox{ran}\left(M_{P}^T\right)
\right|\left|\mbox{ker}\left(M_{P}^T\right)\right|,
\]
hence
\[
z_{C}z_{C^{\perp}}=\frac{q^{\mathfrak{p}}}{q^{\left|V\right|}}
\frac{\left|ker\left(M_{P}^T\right)\right|\left|ker\left(M_{Q}^T\right)\right|}
{\left|ker\left(M^T\right)\right|}
\]
and
\[
\chi=\frac{q^{\left|V\right|}\left|ker\left(M^T\right)\right|}
{\left|ker\left(M_{P}^T\right)\right|\left|ker\left(M_{Q}^T\right)\right|}.
\] 
{From} lemma \ref{lem:prefactor} we now derive that  $\left|ker\left(M^T\right)\right|=q^{c\left(E\right)}$ (with analogous results for $M_{P}^T$ and $M_{Q}^T$).$\hfill\square$

\begin{rem}
\label{rem:num-elem-in-c}As a side remark we note the identities
\begin{equation}
r=q^{\mathfrak{p}}/\left|C_{P}\left(q\right)\right|
\left|C_{P}^{\perp}\left(q\right)\right|\label{eq:schmidtrank}
\end{equation}
and
\[
\left|\mbox{ran}\left(M_{P}^T\right)\right|=\left|C_{P}\left(q\right)\right|
=q^{\left|V\right|}/\left|\mbox{ker}\left(M_{P}^T\right)\right|
=q^{\left|V\right|-c\left(P\right)},
\]
which can be obtained from the proof above.
\end{rem}

\begin{cor}
Considering identity (\ref{eq:lambda-and-c}), we deduce that the Schmidt-rank $\chi$ of a bipartition of edge-qudits $E=A\cup\left(E-A\right)$ and the connectivity function $\lambda$ defined on the graph $G$ satisfy the following equation\[ \chi=q^{\lambda\left(A\right)-1}.\]
\end{cor}

Considering that the matrix $M$ defined in Eq.~\eqref{eq:Mdefinition} is the cycle matroid linked to the states $\ket{\psi_{G}}$ and  $\ket{\varphi_{G}}$ we can now state the following important result,  concluding the second part of the proof.

\begin{cor}
Using the result of lemma \ref{lem:tree-search} to find, by means of the matrix $M$, a branch decomposition of the qudits in the states $\ket{\psi_{G}}$ and $\ket{\varphi_{G}}$, we can efficiently find a subcubic TTN description such that the Schmidt number of all bipartitions following this branch decomposition satisfies
\[
\chi\leq q^{3t-1}.
\]
According to lemma \ref{lem:limited_schmidt_number}, the dimension $D$ of this TTS is limited by $3t-1$ and hence, following lemma \ref{lem:overlap_complexity}, the effort to compute the overlaps $\braket{\alpha}{\psi_{G}}$ and $\braket{\alpha'}{\varphi_{G}}$ grows with at most $O\left(\left|E_{G}\right|t^{3}\right)$. $\hfill\square$
\end{cor}

iii) In this part we want to discuss how to compute the tensor entries in the TTS description of the states $\ket{\psi_{G}}$ and $\ket{\varphi_{G}}$, which we will again denote generically as $\ket{\eta_{G}}$ where no distinction is necessary. The ansatz for the calculation of all tensor elements is the branch decomposition of the edge qudits (concerning the edges of the graph $G=\left(V,E_{G}\right)$) induced by the tree tensor network $T$ with edges $E_{T}$ describing the state. We select an arbitrary edge $e_{0}\in E_{T}$ of the tree to obtain an initial bipartition $E_{G}=P\cup Q$ with $Q=\left(E_{G}-P\right)$ of the edges in $E_{G}$, inducing a bipartition of the set of qudits of the state $\ket{\eta_{G}}$ . We will use the notation $P$ and $Q$ for the edges and the corresponding qudits alike.

Let us consider the Schmidt decomposition belonging to the bipartition. Recalling the proof of lemma \ref{lem:chieins}, the Schmidt coefficients of a decomposition $\ket{\eta_{G}}=\sum_{i}\lambda_{i}\ket{p_{i}}\ket{q_{i}}$, where the states $\ket{p_{i}}$ live on the Hilbert space of the edge qudits in a part $P\subset E_{G}$ and the states $\ket{q_{i}}$ live on the part $Q=E_{G}-P\subset E_{G}$ can be obtained immediately. They are all equal and have the value $\lambda_{i}=r^{-1}=\left| \bar{\mathcal{S}}_{P}\right|/q^{\mathfrak{p}}$. We remember also that there are exactly $r$ of these coefficients. Concerning the Schmidt vectors, we consider the following lemma.

\begin{lem}
A Schmidt basis for a bipartition of the edge qudits $E_{G}=P\cup Q$, $Q=E_{G}-P$ of the state $\ket{\eta_{G}}$ is given by the set of states $\left\{ \ket{p_{i}}\ket{q_{i}}\right\} _{i=1}^{r}$ where
\begin{eqnarray*}
\ket{p_{i}} & := & q^{\left(\mathfrak{p}-\left|V\right|\right)/2}
\sum_{c_{P}\in C_{P}}\ket{c_{P}+\tilde{p}_{i}}\\
\ket{q_{i}} & := & q^{\left(\mathfrak{q}-\left|V\right|\right)/2}
\sum_{c_{Q}\in C_{Q}}\ket{c_{Q}+\tilde{q}_{i}}.
\end{eqnarray*}
Here $\tilde{p_{i}}\in\left(C_{P}^{\perp}\right)^{\perp}$, such that the cosets $\tilde{p_{i}}+C_{P}$ are distinct for different values of $i$, and $C_{P}$ is the cut space of the subspace belonging to the edges belonging to the edge qudits in $P$. ($\left\{ \tilde{q_{i}}\right\} \subset\left(C_{Q}^{\perp}\right)^{\perp}$ is defined analogously; all additions in the kets are modulo $q$).
\end{lem}
\textit{Proof.} We look at the states $\ket{p_{i}}$ first; the states $\ket{q_{i}}$ are treated analogously. The set of states $\left\{ \ket{p_{i}}\right\} $ has to be an orthonormal set which at the same time is a set of eigenstates of the reduced density operator $\rho_{P}=\mbox{Tr}_{Q}\left[\ketbrad{\eta_{G}}\right]$. We define $\bar{\mathcal{S}}_{P}:=\left\{ \mbox{Tr}_{Q}\left[g\right], g\in\mathcal{S}_{P}\right\}$ and recall from the proof of lemma  \ref{lem:chieins} that  $\ket{\eta_{G}}\!\bra{\eta_{G}} =q^{-\mathfrak{c}}\sum_{g\in\mathcal{S}} g$ and hence $\rho_{P}=q^{-\mathfrak{p}}\sum_{g\in\bar{\mathcal{S}}_{P}}g$. Now each $g\in\bar{\mathcal{S}}_{P}$ can be written as  $g=X\left(v\right)Z\left(u\right)$ where $v\in C_{P}$ and $u\in C_{P}^{\perp}$. Applying such an operator to $\ket{p_{i}}$ yields
\begin{eqnarray*}
q^{-\left(\mathfrak{p}-\left|V\right|\right)/2}g\ket{p_{i}} 
& = & X\left(v\right)Z\left(u\right)\sum_{c_{P}\in C_{P}}
\ket{c_{P}+\tilde{p}_{i}}\\
& = & \sum_{c_{P}\in C_{P}}\ket{c_{P}+\tilde{p}_{i}+v}
e^{2\pi iu\cdot \tilde{p}_{i}/q}\\
& = & \sum_{c'_{p}\in C_{P}}\ket{c'_{P}+\tilde{p}_{i}},
\end{eqnarray*}
since $\tilde{p}_{i}\in\left(c_{P}^{\perp}\right)^{\perp}$ and $c_{P}$ is a group. The perpendicularity property of the states $\ket{p_{i}}$ stems from the fact that the vectors $\tilde{p}_{i}$ are from distinct cosets for different values of $i$. We furthermore calculate
\[
\braket{p_{i}}{p_{i}}=q^{\left(\mathfrak{p}-\left|V\right|\right)}
\sum_{c_{P},c'_{P}\in C_{P}}\delta_{c_{P},c'_{P}}=1
\]
following from remark \ref{rem:num-elem-in-c}. The number of Schmidt vectors is indeed $r$, because the number of distinct cosets is, with a  slight generalization of the results of Appendix \ref{sec:NumElProof},  especially Eq.~\eqref{eq:qnccperp}, to $C^{\perp}_{P}$ and  $\left(C^{\perp}_{P}\right)^{\perp}$,
\[
\left|\left(C_{P}^{\perp}\right)^{\perp}/C_{P}\right|
=\left|\left(C_{P}^{\perp}\right)^{\perp}\right|/\left|C_{P}\right|
\underset{\mbox{\tiny{Eq.~\eqref{eq:qnccperp}}}}{=}
\frac{q^{\left|P\right|}}{\left|C_{P}^{\perp}\right|\left|C_{P}\right|}
\underset{\mbox{\tiny{remark \ref{rem:num-elem-in-c}}}}{=}r.
\]
Having proven that the individual states $\ket{\tilde{p}_{i}}$ and $\ket{\tilde{q}_{i}}$ have the given form, we note that the pairing  $\left(\tilde{p}_{i},\tilde{q}_{i}\right)$ for each $i$ is not arbitrary and has to be found out.  In the following we give an algorithm to find these pairs.  We assume in this context that joining the edges of $P$ and $Q$ results in a mere concatenation of the corresponding vectors to simplify the notation. This can always be achieved by a reordering of the edges. The algorithm that we use is as follows

\begin{enumerate}
\item Find the set $\left\{ \tilde{p}_{i}\right\} $, an orthonormal basis $\left\{ \tilde{c}_{Q}\right\} $ of the space $C_{Q}$ and a vector $\left(c_{P}|0\right)\in C$, where $c_{P}\in C_{P}$. \item For each $\tilde{p}_{i}$ find one vector $\left(c_{P}+\tilde{p}_{i}|a_{i}\right)\in C$ where the choice of $a_i$  is in principle arbitrary and just limited by the  set of vectors in $C$. Keep the vectors $a_i$. \item For each vector $a_i$ calculate the corresponding vector $\tilde{q}_{i}:=a_{i}-\sum_{\tilde{c}_{Q}} \tilde{c}_{Q}\left(\tilde{c}_{Q}\cdot a_{i}\right).$
\end{enumerate}
By construction, the vectors $\tilde{q}_{i}$ are all elements of $\left(C_{Q}^{\perp}\right)^{\perp}$. Furthermore we notice that there are efficient algorithms for all these steps. $\hfill\square$

This bipartition enables us to compute all tensor entries efficiently.

Consider that using the TTN description of the state $\ket{\eta_{G}}$ the states $\ket{p_{i}}$ and $\ket{q_{i}}$ can be written as
\[
\ket{p_{i}}=\sum_{jk}A_{ijk}^{[P]}\ket{j}_{P_{1}}\ket{k}_{P_{2}},
\quad\ket{p_{i}}=\sum_{jk}A_{ilm}^{[Q]}\ket{l}_{Q_{1}}\ket{m}_{Q_{2}}
\]
with suitable tensors $A^{[P]}$ and $A^{[Q]}$ and states $\ket{j}_{P_{1}},\ket{k}_{P_{2}}$. The states $\ket{j}_{P_{1}}$ are living on the Hilbert space $P_{1}$ belonging to the leaves (and hence to the corresponding qudits) that are part of the sub-tree of $T$ attached to the tensor $A^{[P]}$ by its index $i$. Also the states $\ket{j}_{P_{1}}$ can be written as $\ket{j}_{P_{1}}=\sum_{rs}A_{jrs}^{[P_{1}]}\ket{r}_{P_{11}}\ket{s}_{P_{12}}$ (analogous arguments apply to the states $\ket{k}_{P_{2}},\ket{l}_{Q_{1}},\ket{m}_{Q_{2}}$.) To be able to compute the entries of the tensors we hence need the states belonging to sub-trees which can be derived from the initial Schmidt decomposition.

\begin{lem}
\label{lem:state-division}Let $\ket{i}_{E}$ be a state on the qudits corresponding to a set of edges $E$, defined as  $\ket{i}_{E}=\sum_{c_{E}\in C_{E}}\ket{c_{E}+d\left(i\right)}$, where $C_{E}$ is the cut space of the incidence matrix of the graph $G=\left(V,E\right)$ belonging to the qudits as defined above. Let $E=P\cup Q$, $Q=E-P$ be a bipartition of the qudits. Then
\[
\ket{i}_{E}=\ket{P_{i}}\ket{Q_{i}},
\]
where
\begin{eqnarray*}
\ket{P_{i}} & = & \sum_{c_{P}\in C_{P}}\ket{c_{P}+d\left(i\right)_{P}}\\
\ket{Q_{i}} & = & \sum_{c_{Q}\in C_{Q}}\ket{c_{Q}+d\left(i\right)_{Q}}.
\end{eqnarray*}
The states $\ket{P_{i}}$ and $\ket{Q_{i}}$ live on the Hilbert spaces of $P$ and $Q$ respectively and the vectors $d\left(i\right)_{P}$ and $d\left(i\right)_{Q}$ are the parts of the vector $d\left(i\right)$ belonging to the respective qudits.
\end{lem}
\textit{Proof.} A reordering of the position of the qudits in $\ket{i}_{E}$, so that the merging of the vectors  $c_{P},c_{Q},d\left(i\right)_{P}$ and $d\left(i\right)_{Q}$ becomes a  concatenation yields
\begin{eqnarray*}
\ket{i}_{E} & = & 
\sum_{c_{P}\in C_{P},c_{Q}\in C_{Q}}
\ket{\left(c_{P}|c_{Q}\right)+\left(d\left(i\right)_{P}
|d\left(i\right)_{Q}\right)}\\
& = & \sum_{c_{P}\in C_{P},c_{Q}\in C_{Q}}
\ket{\left(c_{P}|c_{Q}\right)+d\left(i\right)}.
\end{eqnarray*}
The sets of vectors $\left\{ c_{E}\right\} ,\left\{ c_{P}\right\} $ and $\left\{ c_{Q}\right\} $ are the images of the matrices $M_{E}^T,M_{P}^T$ and $M_{Q}^T$ of identity \ref{eq:Mdefinition} respectively, where the index denotes the edges that the columns correspond to. Since $M_{E}=\left(M_{P}|M_{Q}\right)$, we obtain the result that for each $c_{E}$ there is exactly one pair $\left(c_{P},c_{Q}\right)$ immediately.$\hfill\square$

We observe that the involved sets $C_{P},C_{Q}$ and the vectors  $d\left(i\right)_{P},d\left(i\right)_{Q}$ can be found efficiently. Now we write 
$\ket{i}_{E}=\sum_{ijk}A_{ijk}\ket{P_{j}}\ket{Q_{k}}$ and deduce that 
\[
A_{ijk}=\delta_{ij}\delta_{ik},
\]
except for $A^{[0]}$ and $A^{[1]}$ which have to absorb the square root of the Schmidt coefficients also and hence $A_{ijk}^{[0,1]}=r^{-1/2}\delta_{ij}\delta_{ik}$.

This concludes the proof of theorem \ref{thm:scaling}. $\hfill\square$

\subsection{The bipartite entanglement of the general encoding schemes  (e.g.\ GHZ-product state and the vertex model state)}

So far we have only considered the computational complexity using an encoding into the states $\ket{\psi_{G}}$ and $\ket{\varphi_{G}}$. In this section we want to extend the efficiency statement to the alternative encoding schemes discussed in section \ref{sec:Extending-the-formalism}.

The major modification leading to these schemes and complicating the entanglement aspect is the extension of measurements from one qudit to two or more. In a branch decomposition, the sites being involved in these measurements have to be placed in their own sub-trees, which we will refer to as \emph{contraction sites}. The contraction of the highly entangled states $\ket{\alpha_{e}}$ with these contraction sites will in general not be efficient, but since the size of the contraction sites is limited, this only leads to a constant computational overhead. In a branch decomposition of a state of the extended encoding schemes we can represent the contraction sites as leaves.

The remaining question is ``What is the entanglement of bi-partitions in a branch decomposition where the contraction sites are leaves''? By construction, we immediately find that this question can be answered by looking at the number of states (in our schemes those are either q-dimensional Bell pairs or GHZ states) that are shared by different contraction sites and cut by the branch decomposition. See also  Figs.~\ref{fig:ghzscheme} and \ref{fig:vertexscheme}.

Once we contract the contraction sites in the graphical representation of the general picture (like given in Fig.~\ref{fig:ghzscheme}) to single vertices, we obtain a new graph where the edges represent Bell pairs shared by contraction sites. Graph theory immediately tells us that also in this case the tree width is the decisive quantity of the (contracted) graph that governs the minimum number of \emph{states}  (and hence ebits) that have to be cut in a branch decomposition. The tree width of the contracted graph is carried over from the underlying graph of the classical interaction graph. Thus we can conclude that theorem \ref{thm:scaling} applies for the alternative encoding schemes as well, and the decisive parameters can be derived immediately from the respective encoding patterns.

We also emphasize that non-planar graphs of logarithmically bounded tree-width, as well as non-local interactions are covered by this result. Results regarding efficient computation of homogeneous Potts model partition functions on graphs of logarithmically bounded tree-width have been obtained before, though with entirely different methods. We emphasize that our approach, in contrast to previous approaches, can handle without difficulty also inhomogeneous models. Moreover, it leaves a lot of space for generalizations.


\begin{thebibliography}{10}

\bibitem{Wu84}F.Y.\ Wu, Rev.\ Mod.\ Phys.\ \textbf{54}, 235 (1964).

\bibitem{Ba82}F.\ Barahona, J.\ Phys.\ A \textbf{15}, 3241 (1982).

\bibitem{SDV06}Y.\ Shi, L.\ Duan and G.\ Vidal, Phys.\ Rev.\ A \textbf{74}, 022320 (2006).

\bibitem{VDB07}M.\ Van den Nest, W.\ Dür, G.\ Vidal, H.\ J.\ Briegel, Phys.\ Rev.\ A \textbf{75}, 012337 (2007).

\bibitem{Go97}D.\ Gottesmann, PhD-thesis, Caltech 1997.

\bibitem{HDM05}E.\ Hostenes, J.\ Dehaene, B.\ De Moor, Phys.\ Rev.\ A \textbf{71}, 042315 (2005).

\bibitem{VPhD}M.\ Van den Nest, PhD thesis, K.U.Leuven 2005.

\bibitem{Hei05}M.\ Hein, W.\ Dür, J.\ Eisert, R.\ Raussendorf, M.\ Van den Nest, H.J.\ Briegel, \emph{Proceedings of the International School of Physics ``Enrico Fermi'' on Quantum Computers, Algorithms and Chaos}, Varenna, Italy, July (2005); M.\ Hein, J.\ Eisert, and H.\ J.\ Briegel, Phys.\ Rev.\ A \textbf{69}, 062311 (2004).

\bibitem{SW00}D.\ Schlingemann, R.F.\ Werner, Phys.\ Rev.\ A \textbf{65}, 012308 (2001).

\bibitem{RB01}R.\ Raussendorf, H.J.\ Briegel, Phys.\ Rev.\ Lett.\ \textbf{86}, 5188 (2001).

\bibitem{VDB06}M.\ Van den Nest, W.\ Dür, H.J.\ Briegel, Phys.\ Rev.\ Lett.\ \textbf{98}, 117207 (2007).

\bibitem{many}D.\ A.\ Lidar and O.\ Biham, Phys.\ Rev.\ E \textbf{56}, 3661 (1997); D.\ A.\ Lidar, New J.\ Phys.\ \textbf{6}, 167 (2004);  R.\ D.\ Somma, C.\ D.\ Batista, and G.\ Ortiz, Phys.\ Rev.\ Lett.\ \textbf{99}, 030603 (2007); H.\ Bombin and M.A.\ Martín-Delgado, Phys. Rev. A \textbf{77}, 042322 (2008); D.\ Aharonov, I.\ Arad, E.\ Eban, Z.\ Landau, e-print arXiv:quant-ph/0702008; I.\ Arad, Z.\ Landau, e-print arXiv:0805.0040; V.\ Murg, F.\ Verstraete, I.\ Cirac, Phys.\ Rev.\ Lett.\ \textbf{95}, 057206 (2005); F.\ Verstraete, M.M.\ Wolf, D.\ Perez-Garcia, J.I.\ Cirac, Phys.\ Rev.\ Lett.\ 96, 220601 (2006); J.\ Geraci, Quant. Inf. Proc. \textbf{7} 227, (2008); J.\ Geraci, D.A.\ Lidar, Comm.\ Math.\ Phys.\ \textbf{279}, 735 (2008).

\bibitem{BR06}S.\ Bravyi, R.\ Raussendorf, Phys.\ Rev.\ A \textbf{76}, 022304 (2007).

\bibitem{MS06}I.\ Markov, Y.-Y.\ Shi, SIAM Journal on Computing, \textbf{38}(3): 963 (2008).

\bibitem{CDBip}M.\ Van den Nest, W.\ Dür, H.J.\ Briegel, Phys.\ Rev.\ Lett.\ \textbf{100}, 110501 (2008); G.\ De las Cuevas, W.\ Dür, M.\ Van den Nest, H.J.\ Briegel, J.\ Stat.\ Mech.\ P07001 (2009).

\bibitem{GR01}C.\ Godsil, G.\ Royle, \emph{``Algebraic Graph Theory''}, Springer New York - Berlin - Heidelberg (2001).

\bibitem{ME86} A.C.N.\ de Magalh\~aes, J.W.\ Essam, J.Phys.A: Math.\ Gen.\ \textbf{19} 1655 (1986).

\bibitem{Ki03}A.\ Kitaev, Annals Phys.\ \textbf{303}, 2 (2003).

\bibitem{BR01}H.\ J.\ Briegel and R.\ Raussendorf, Phys.\ Rev.\ Lett.\ \textbf{86}, 910 (2001).

\bibitem{RBB03}R.\ Raussendorf, D.E.\ Browne, H.J.\ Briegel, PRA \textbf{68}, 022312 (2003).

\bibitem{Va06a}M.\ Van den Nest, A.\ Miyake, W.\ D\"ur and H.\ J.\ Briegel, Phys.\ Rev.\ Lett.\ \textbf{97}, 150504 (2006); M.\ Van den Nest, W.\ D\"ur, A.\ Miyake and H.\ J.\ Briegel, New J.\ Phys.\ \textbf{9}, 204 (2007)

\bibitem{GE07}D.\ Gross, J.\ Eisert, Phys.\ Rev.\ Lett.\ \textbf{98}, 220503 (2007); D.\ Gross, J.\ Eisert, N.\ Schuch, D.\ Perez-Garcia, Phys.\ Rev.\ A \textbf{76}, 052315 (2007); D.\ Gross, J.\ Eisert, e-print arXiv:0810.2542.

\bibitem{Brprep}H.J.\ Briegel, D.E.\ Browne, W.\ D\"ur, R.\ Raussendorf, M.\ Van den Nest, Nature Physics \textbf{5}, 19 (2009).

\bibitem{NGESP06}M.A.\ Nielsen, Reports On Mathematical Physics, \textbf{57} 1, 147 (2006).

\bibitem{Moprep}M.J.\ Bremner, C.\ Mora, A.\ Winter, e-print arXiv:0812.3001

\bibitem{HW06}P.\ Hlineny, G.\ Whittle, Matroid tree-width.\ Eur.\ J.\ Comb.\ \textbf{27}(7), 1117 (2006).

\bibitem{OPhD}S.\ Oum, PhD thesis, Princeton University 2005.

\bibitem{RS84}N.\ Robertson, P.D.\ Seymore, J.\ Combin.\ Theory Ser.\ B, \textbf{36}(1), 49 (1984).

\bibitem{footnote}Please note that this definition differs from the one given in Ref.~\cite{VDB06}. In the present paper, the factor $q^{\kappa}$ has been absorbed into $\ket{\psi_{G}}$.

\end{thebibliography}
\end{document}